\documentclass{siamltex1213}

\pdfoutput=1

\usepackage[cmex10]{amsmath}
\usepackage{amssymb}
\usepackage{mathtools}
\usepackage{stmaryrd}
\usepackage{url}
\usepackage{tikz}

\newtheorem{rem}{Remark}
\newtheorem{notdef}[theorem]{Definition and Notation}

\newcommand{\jp}[1]{\llbracket #1 \rrbracket} 
\newcommand{\av}[1]{\{ #1 \}}
\newcommand{\corr}{\text{corr}}
\newcommand{\ucorr}{\ensuremath{u^{\corr}}}
\newcommand{\uinf}{\ensuremath{u^{\infty}}}
\newcommand{\scorr}{\ensuremath{\sigma^{\corr}}}
\newcommand{\sinf}{\ensuremath{\sigma^{\infty}}}
\newcommand{\Omegainf}{\ensuremath{\Omega^{\infty}}}
\def\clap#1{\hbox to 0pt{\hss#1\hss}}

\def\mathclap{\mathpalette\mathclapinternal}

\def\mathclapinternal#1#2{\clap{$\mathsurround=0pt#1{#2}$}}
\def\uscore{\protect\rule{.2em}{.2pt}}
\def\Uscore{\protect\rule{.35em}{.2pt}}
\newcommand{\segres}[2]{\textit{seg\uscore{#1}\Uscore{}res\uscore{#2}}}
\newcommand{\segresR}[3]{\textit{seg\uscore{#1}\Uscore{}res\uscore{#2}\Uscore{}r{#3}}}
\newcommand{\ci}[3]{\textit{{#1}CI\Uscore{#2}\Uscore{#3}}}
\newcommand{\lnMAGj}{\ensuremath{lnMAG_{\vec{j},loc}\,}}
\newcommand{\totDIFFj}{\ensuremath{totDIFF_{\vec{j},loc}\,}}

\begin{document}
\title{A Discontinuous Galerkin Method to Solve the EEG Forward Problem Using the Subtraction Approach}
\author{\text{Christian Engwer \footnotemark[1] \footnotemark[2] \footnotemark[3] \footnotemark[4]}
\and
\text{Johannes Vorwerk \footnotemark[2] \footnotemark[5] \footnotemark[7] \footnotemark[6]}
\and
\text{Jakob Ludewig \footnotemark[3]  \footnotemark[5] \footnotemark[6]}
\and
\text{Carsten H. Wolters \footnotemark[5] \footnotemark[6]}
}


\renewcommand{\thefootnote}{\fnsymbol{footnote}}
\footnotetext[1]{Corresponding author: \ christian.engwer@uni-muenster.de}
\footnotetext[2]{The first two authors contributed equally to this work.}
\footnotetext[3]{Institute for Computational and Applied Mathematics,
  University of M\"unster, Einsteinstra\ss e 62, 48149 M\"unster,
  Germany}
\footnotetext[4]{Cluster of Excellence EXC 1003, Cells in Motion, CiM, M\"unster, Germany}
\footnotetext[5]{Institute for Biomagnetism and
  Biosignalanalysis, University of M\"unster, Malmedyweg 15, 48149 M\"unster, Germany
}
\footnotetext[7]{Scientific Computing and Imaging (SCI) Institute, University of Utah, 72 S Central Campus Drive, Salt Lake City, UT-84112, USA}
\footnotetext[6]{This work was partially supported by the Deutsche Forschungsgemeinschaft
(DFG), project WO1425/7-1, the Priority
  Program 1665 of the Deutsche Forschungsgemeinschaft (DFG) (WO1425/5-1, 5-2), the
  Cluster of Excellence 1003 of the Deutsche Forschungsgemeinschaft (DFG EXC 1003 Cells in Motion),
  and by EU project ChildBrain (Marie Curie Innovative
  Training Networks, grant agreement no. 641652).}
\maketitle

\begin{abstract}
  In order to perform electroencephalography (EEG) source
  reconstruction, i.e., to localize the sources underlying a measured
  EEG,
  the electric potential distribution at the electrodes generated
  by a dipolar current source in the brain has to be
  simulated, which is
  the so-called EEG forward problem.
  To solve it accurately, it is necessary to
  apply numerical methods that are able to take the individual
  geometry and conductivity distribution of the subject's head into account.
  In this context, the finite element method (FEM) has shown high numerical accuracy with the
  possibility to model complex geometries and conductive
  features, e.g., white matter conductivity anisotropy.
  In this article, we introduce and analyze the application of a
  discontinuous Galerkin (DG) method, a finite element method that
  includes features of the finite volume framework, to the EEG forward
  problem. The DG-FEM approach fulfills the conservation
  property of electric charge also in the discrete case, making it
  attractive for a variety of applications.
  Furthermore, as we show, this approach
  can alleviate modeling inaccuracies that might occur in
  head geometries when using classical FE methods, e.g., so-called ``skull leakage effects'', which may occur in areas where the thickness of the skull is in the range of the mesh resolution.
  Therefore, we derive a DG formulation of the FEM
  subtraction approach for the EEG forward problem and present numerical results that highlight the
  advantageous features and the potential benefits of the proposed approach.
\end{abstract}



\section{Introduction}
EEG source reconstruction is nowadays widely used in both research
and clinical routine to investigate the activity of the human brain,
as it is a non-invasive, easy to perform, and relatively cheap technique \cite{MHama1,CHW:Bre2012}.
To reconstruct the active brain areas from the electric potentials measured
at the head surface, it is necessary
to simulate the electric potential generated by a dipolar current source in the gray matter compartment of the brain, the so-called
EEG forward problem. The achievable accuracy in solving the forward problem
strongly depends on a realistic modeling of shape and conductive features of
the volume conductor, i.e., the human head. Therefore, it is necessary to
apply numerical methods to solve the underlying partial differential
equations in realistic geometries, since analytical solutions exist for only few special cases, e.g., nested shells  \cite{CHW:Mun93}. Different numerical methods
have been proposed to solve this problem, e.g., boundary element methods (BEM) \cite{CHW:Mos99,CHW:Aka2010,AGram2011,CHW:Sten2012},
finite volume methods (FVM) \cite{MCook2006}, finite difference methods (FDM) \cite{CHW:Wen2008,CHW:Vat2009,CHW:Mon2014}, or
finite element methods (FEM) \cite{CHW:Ber91,CHW:Mar98,CHW:Sch2002,FDrec1,CHW:Ram2006,CHW:Pur2011}.
Finite element methods were shown to achieve high numerical accuracies \cite{FDrec1,JVorw2012}
and offer the important possibility to model complex
geometries and also anisotropic conductivities, with only a weak influence on the computational effort
\mbox{\cite{JVorw2014}}. The computational burden of using FE methods to solve the EEG forward problem could be clearly reduced by the introduction of transfer
 matrices and fast solver methods \mbox{\cite{CHW:Wei2000,CHW:Gen2004,CWolt2004}}.

One of the main tasks in applying FE methods to solve the EEG forward problem is to
deal with the strong singularity introduced by the source model of a
current dipole. Therefore, different approaches to solve the EEG
forward problem using the FEM have been proposed, e.g., the Saint-Venant \cite{CHW:Tou65,CHW:Sch94,HBuch1997,JVorw2012},
the partial integration \cite{YYan1991,CHW:Wei2000,SVall2010,JVorw2012}, the Whitney or Raviart-Thomas \cite{CHW:Tan2005,CHW:Pur2011}, or the subtraction approach \cite{CHW:Ber91,CHW:Mar98,CHW:Sch2002,CWolt1,FDrec1,JVorw2012}.
All these approaches rely on a continuous Galerkin FEM (CG-FEM) formulation, also called Lagrange or conforming FEM, i.e., the resulting solution for the electric potential is
continuous.

The use of tetrahedral \cite{CHW:Mar98,FDrec1,JVorw2014} as well as that of hexahedral \cite{CHW:Sch2002,MRull2009,UAydi2014,CHW:Ayd2015} meshes
has been proposed for solving the EEG forward problem with the FEM. Tetrahedral meshes can be generated by
constrained Delaunay tetrahedralizations (CDT) from given tissue surface
representations \cite{FDrec1,JVorw2014}. This approach has the advantage that smooth tissue surfaces are well represented in the model.
On the other side, the generation of such models is difficult in practice and might cause unrealistic model features, e.g.,
holes in tissue compartments such as the foramen magnum and the optic canals in the skull are often artificially
closed to allow CDT meshing. Furthermore, CDT modeling necessitates the generation of nested, non-intersecting, and -touching surfaces. However, in reality, surfaces might touch, for example, the inner skull and outer brain surface. Hexahedral models do
not suffer from such limitations, can be easily generated from voxel-based magnetic resonance imaging (MRI) data,
and are more and more frequently used in source analysis applications \cite{MRull2009,UAydi2014,CHW:Ayd2015}.
This paper therefore focuses on the application of FE methods with hexahedral meshes. However, the application of the CG-FEM with hexahedral meshes has the disadvantage
that the representation of thin tissue structures in combination with insufficient mesh resolutions might result in
geometry approximation errors. It has been shown, e.g., in \cite{HSonn2013}, that the combination
of thin skull structures and insufficient hexahedral mesh resolutions might result in so-called skull leakages
in areas where scalp and CSF elements are erroneously connected via single skull vertices or edges,
as illustrated in Figure \ref{fig:shortcut}. Such leakages can lead to significantly inaccurate results when using
vertex-based methods like, e.g., the CG-FEM, and might be one of the main reasons why in a recent
head modeling comparison study for EEG source analysis in presurgical epilepsy diagnosis, the use of the CG-FEM with a four-layer hexahedral head model with a resolution of 2 mm did not lead to better results than those for simpler head models, i.e., a three-layer local sphere and a three-layer boundary element head model \cite{Birot2014}.

In this paper, we derive the mathematical equations underlying
the forward problem of EEG and introduce its solution using the
subtraction approach. After a short explanation of the strengths and
weaknesses of this approach, we propose and evaluate a
new formulation of the subtraction approach on the basis of the discontinuous
Galerkin FEM (DG-FEM). We then show that, although the CG- and DG-FEM
achieve similar numerical accuracies in multi-layer sphere validation studies
with high mesh resolutions, the DG-FEM mitigates the problem of skull leakages
in case of lower mesh resolutions. The results of the sphere studies are complemented and underlined by the results obtained in a realistic six-compartment head model.

\section{Theory}

\subsection{The forward problem}

The partial differential equation underlying the EEG
forward problem can be derived by introducing the quasi-static
approximation of Maxwell's equations \cite{MHama1,CHW:Bre2012}. When relating
the electric field to a scalar potential, $E = - \nabla u$, and
splitting up the current density $J$ into a term $f$, which describes
the current source,
and a return current, or flux,
$- \sigma \nabla u$ with  $\sigma (x)$ being the conductivity distribution in the head domain, we obtain a \emph{Poisson equation}
\begin{subequations}
  \label{eq:forward-strong}
\begin{align}
  - \nabla \cdot ( \sigma \nabla u ) &= f \quad &\mbox{in}\; \Omega,\label{eq:forward-strong1}\\
    \sigma \partial_{\mathbf{n}}  u &= 0\quad &\mbox{on}\; \partial \Omega ,\label{eq:forward-strong2}
\end{align}
\end{subequations}
where $\Omega$ denotes the head domain, which is assumed to be open
and connected, and $\partial \Omega$ its boundary. We have homogeneous Neumann
boundary conditions here, since we assume a conductivity $\sigma (x) =
0$ for all $x \notin \bar{\Omega}$.

\subsection{The subtraction approach}
\label{sec:subtraction}
We briefly derive the classical subtraction FE
approach as presented in \cite{CWolt1, FDrec1}.
We assume the commonly used point-like dipole source at position
$y$ with moment $p$, $f_y ( x ) = \nabla \cdot \left(p  \delta_{y} ( x )\right)$.
This choice complicates the further mathematical treatment, as the right-hand side is not
square-integrable in this case.
However, when assuming that there exists a non-empty open
neighborhood $\Omegainf$ of the source position $y$ with
constant isotropic conductivity $\sinf$, we can split the potential $u$ and the conductivity $\sigma$ into two parts:
\begin{subequations}
\begin{align}
 u &= \uinf + \ucorr, \\
 \sigma &= \sinf + \scorr \label{eq:sigma}.
\end{align}
\end{subequations}
$\uinf$ is the potential in an unbounded, homogeneous conductor
and can be calculated analytically:
$\uinf (x)= \frac{1}{4 \pi\sinf}\frac{\langle p,x-y\rangle}{ | x-y |^3}$.
The more general case of anisotropic conductivities can be treated, too \cite{CWolt1, FDrec1}, but is not especially derived
here.

Inserting the decomposition of $u$ into \eqref{eq:forward-strong}
and subtracting the homogeneous solution, again
results in a Poisson equation for the searched correction potential $\ucorr$:
\begin{subequations}
  \label{eq:subtraction-strong}
\begin{align}
    \label{eq:subtraction-strong1}
    - \nabla \cdot ( \sigma \nabla \ucorr) &= \nabla \cdot ( \scorr \nabla \uinf ) \quad &\text{in }\Omega,\\
    \label{eq:subtraction-strong2}
    \sigma \partial_{\mathbf{n}} \ucorr & = - \sigma \partial_{\mathbf{n}} \uinf \quad &\text{on }\partial \Omega.
\end{align}
\end{subequations}
To solve this problem numerically, \cite{FDrec1} propose a conforming
first-order finite element method: Find $\ucorr \in V_h \subset
H^1$ such that it fulfills the weak formulation
\begin{equation}
  \label{eq:subtraction-weak}
  \int\limits_\Omega \sigma \nabla \ucorr \cdot \nabla v dx =
  - \int\limits_\Omega \scorr \nabla \uinf \cdot \nabla v dx
  - \int\limits_{\partial \Omega} \sinf \partial_{\mathbf{n}} \uinf  v ds.
\end{equation}
The weak form can be heuristically derived by
multiplication with a test function $v \in V_h$ and subsequent partial
integration. Reorganization of some terms and applying the identity
\eqref{eq:sigma} yields the proposed form in equation  \eqref{eq:subtraction-weak}.
The subtraction approach is theoretically well understood. The existence
of a solution as well as the uniqueness and convergence of this solution are
examined in \cite{CWolt1, FDrec1}.

\begin{figure}%
\centering
\raisebox{-0.5\height}{\includegraphics[width=0.4\textwidth]{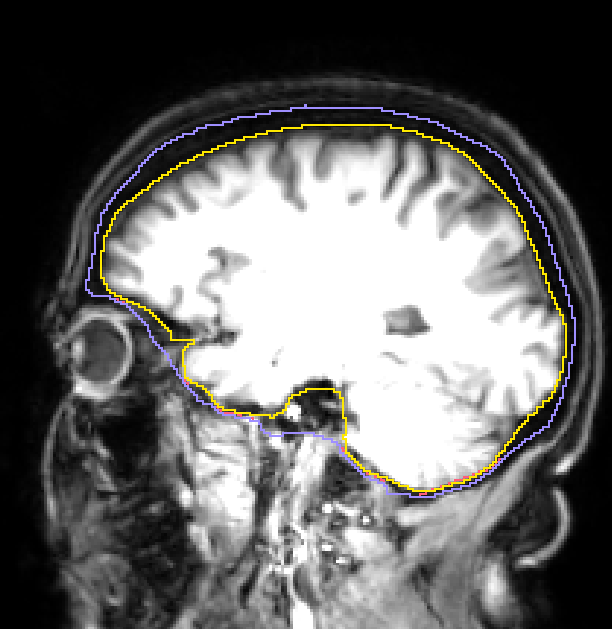}}\hfill
\raisebox{-0.5\height}{\includegraphics[width=0.55\textwidth]{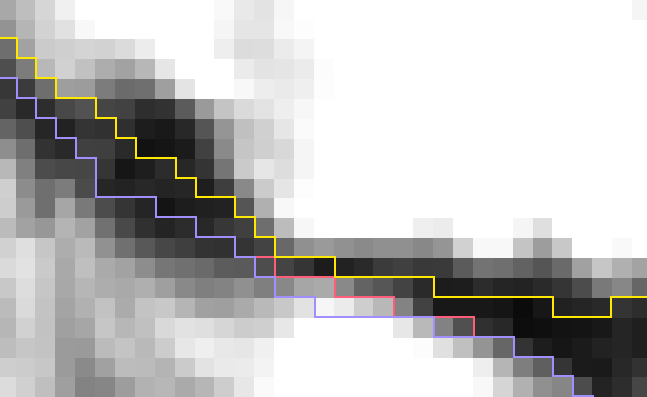}}%
\caption{Sketch of segmentation that might lead to leakage effects (left). The yellow line shows the inner skull surface,
the red line the original outer skull surface, the blue line the corrected outer skull surface. Where the red and the blue line overlap, only the blue line is visible.
In the magnified detail, the scalp and CSF show two erroneous connections via single vertices or edges
(right subfigure, where the red and the yellow line touch each other). Such a segmentation can lead to
  significantly inaccurate results when using vertex-based methods such as, for example,
  the CG-FEM.}%
\label{fig:shortcut}%
\end{figure}

\subsection{Skull leakage effects}
\label{sec:skullleakage}
 As discussed in the introduction, hexahedral meshes are frequently used in practical applications of FEM-based EEG/MEG source analysis,
  due to the clearly simplified creation process in comparison to CDT meshes.
  A pitfall that has to be taken into account in this scenario is leakage effects, especially in the thin skull compartment.
  If the segmentation resolution, i.e., the resolution of the discrete approximation of the geometry, is coarse compared to the thickness of the skull,
  segmentation artifacts as illustrated in Figure \ref{fig:shortcut} (yellow and red lines)
  occur.
  When directly generating a hexahedral mesh from this segmentation, elements belonging to the
  highly conductive compartments interior to the skull, i.e., most often the CSF, and to
  the skin compartment are now connected via a shared vertex or edge, although they are physically separated in reality.
  When using such a mesh for, e.g., the CG-FEM with Lagrange Ansatz-functions,
   these artifacts lead to skull leakage, as sketched in
  Figure \ref{fig:leakage}.
  As a consequence of the vertex-based Ansatz-functions, the shared vertices have inadequately
  high entries in the stiffness matrix, which result in current leakage ``through'' these vertices.

This effect remains unchanged even when (globally or locally) refining the resolution of the mesh. An increase of the image -- and thereby also the segmentation -- resolution might eliminate this effect, but is usually not possible.
  Instead, this problem might be circumvented by artificially increasing the thickness of the skull segmentation
  in these areas (blue line in Figure \ref{fig:shortcut}). However, this workaround might, again, lead to inaccuracies in the EEG forward computation due to the now too thick representation of the skull compartment.

\begin{figure}%
\centering
\includegraphics[width=0.3\textwidth]{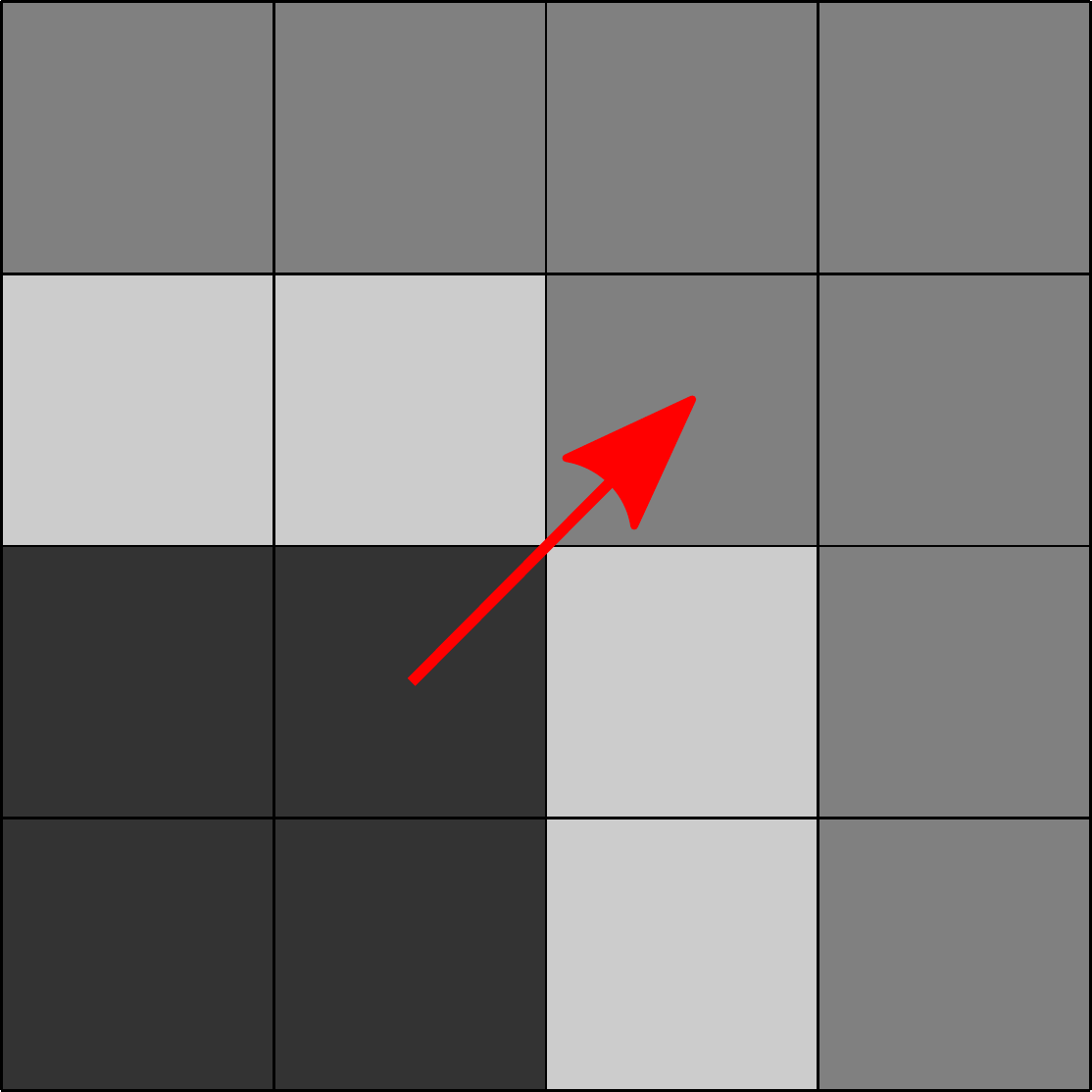}\qquad\qquad
\includegraphics[width=0.3\textwidth]{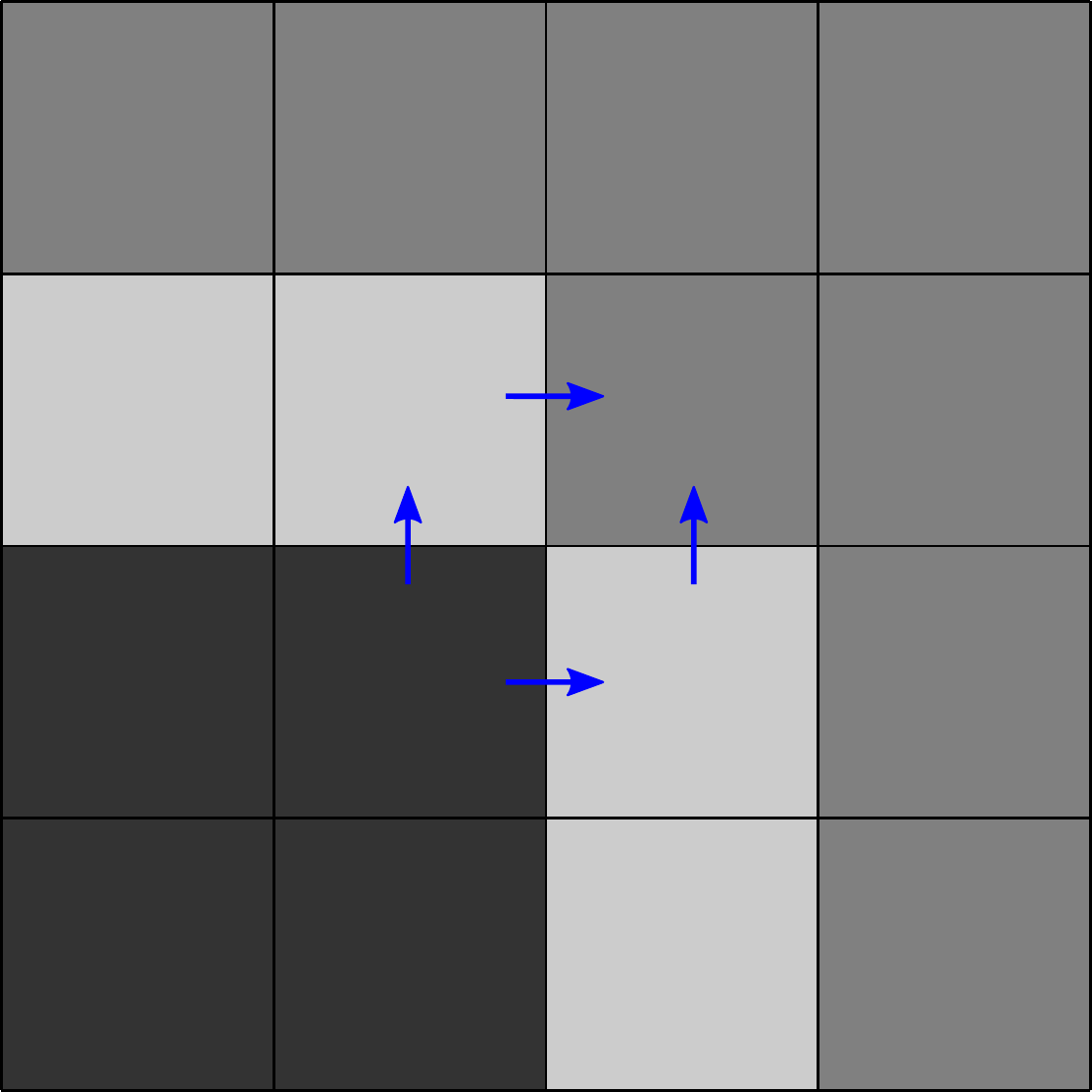}
\\
\parbox{0.3\textwidth}{\centering CG-FEM}\qquad\qquad
\parbox{0.3\textwidth}{\centering DG-FEM}
\caption{CG-FEM simulations lead to an overestimated
  electric current at degenerated vertices of the skull. This effect is due
  to the vertex-based discretization, which considers only the potential,
  but not the electric current. The DG-FEM is based on a current
  reconstruction through cell faces. Therefore, these methods do not
  overestimate the electric current, even in the presence of segmentation artifacts.}%
\label{fig:leakage}%
\end{figure}

  In the following section, we derive a discontinuous Galerkin (DG) formulation for the subtraction FE approach. This formulation
  has the advantage that it is locally charge preserving and
  controls the current flow through element faces,
thereby preventing possible leakage effects; see illustration in Figure
  \ref{fig:leakage}.

\subsection{A discontinuous Galerkin formulation}%
Preserving fundamental physical properties is very important
in order to obtain reliable simulation results.
As discussed in the previous section, a correct approximation of the
electric current is crucial for reliable simulation results. Continuity of the
normal component of the current directly implies conservation of
charge.

The discontinuous Galerkin method allows to construct formulations
that preserve such conservation properties
also in the discretized space.
We first discuss
which quantities to preserve when using the subtraction approach for the continuous problem
and then introduce a discontinuous Galerkin formulation.

\subsubsection{Conservation properties}
\label{sec:conservation}%
A fundamental physical property is the conservation of
charge:
\begin{equation}
  \label{eq:continuity1}
  \int\limits_{\partial K} \sigma \nabla u \cdot \vec{n}\, ds = \int\limits_K f_y dx,
\end{equation}
for any control volume $K\subseteq \Omega$.
Following the subtraction
approach, we split the current $\sigma \nabla u = (\sinf + \scorr)
\nabla (\uinf + \ucorr)$. Rearrangement then yields
\begin{equation*}
    \int\limits_{\partial K} \sigma \nabla \ucorr \cdot \vec{n}\, ds =
    - \int\limits_{\partial K} \scorr \nabla \uinf \cdot \vec{n}\, dx
    \underbrace{- \int\limits_{\partial{K}} \sinf \nabla \uinf \cdot
     \vec{n}\,ds
      + \int\limits_K f_y dx}_{\equiv0}\,.
\end{equation*}
Applying Gauss's theorem to the right-hand side, we obtain a conservation
property for the correction potential
\begin{equation}
  \label{eq:continuity2}
    \int\limits_{\partial K} \underbrace{\sigma \nabla \ucorr}_{\vec{j}^\corr} \cdot \vec{n}\, ds =
    \int\limits_K \underbrace{- \nabla \cdot \scorr \nabla \uinf}_{f^\corr} dx\,,
\end{equation}
which basically states that the correction potential $\ucorr$ causes a flux $\vec{j}^\corr$; the charge corresponding to this flux is a conserved property with
source term $f^\corr = \nabla \cdot \scorr \nabla \uinf$.

For FE methods this property carries over to the discrete
solution, if the test space contains the characteristic function,
which is one on $K$ and zero everywhere else. In general, a conforming
discretization does not guarantee this property.

Conservation of charge also holds for $\uinf$ in the case of a homogeneous volume conductor (with conductivity $\sinf$ in our case).
Thus, the normal components of both the electrical flux $\sigma \nabla u$ and $\sinf
\nabla \uinf$ are continuous. Rewriting $\vec{j}$ in terms of
$\scorr$, $\sinf$, $\ucorr$ and $\uinf$ we can show that the normal component of
$\sigma \nabla \ucorr + \scorr \nabla \uinf$ is also continuous.

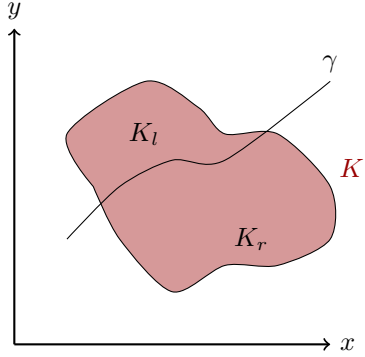
\begin{figure}
  \centering
  \begin{center}
    \begin{tikzpicture}[scale = 0.7]
      \draw [->, thick] (0,0) coordinate (origin) -- (0,6) node
      [above] {$y$}; \draw [->, thick] (origin) -- (6,0) node
      [right] {$x$}; \draw [fill=red!60!black!40]
      plot[smooth]coordinates{(1.5,3)(2,2)(3,1)(4,1.5)(5,1.5)(6,2)(6,3)(5,4)(4,4)(3.5,4.5)(2.5,5)(1,4)(1.5,3)};
      \draw
      plot[smooth]coordinates{(1,2)(2,3)(3,3.5)(4,3.5)(6,5)}node
      [above]{$\gamma$}; \draw (4,2) node[right]{$K_r$};
      \draw (2,4) node[right]{$K_l$}; \draw(6,3) node[above
      right,red!60!black!100]{$K$};
    \end{tikzpicture}
  \end{center}

  \caption{Interface $\gamma$ splits $K$ into two parts.}
  \label{fig:jump}
\end{figure}

\begin{definition}
  \label{def:jump}
  We consider an arbitrary interface $\gamma$, which separates the
  control volume $K$ into two patches $K_l$ and $K_r$ (see Figure
  \ref{fig:jump}).
  Following \mbox{\cite{Arnold2002}}, we introduce the  jump of a
  scalar function $u$ or a vector-valued function $\vec{v}$ along $\gamma$
  as
    \begin{subequations}
  \begin{align}
      \jp{u}_\gamma &\coloneqq
    u \vert_{\partial K_l} \,\vec{n}_{K_l}\; +\;
    u \vert_{\partial K_r} \,\vec{n}_{K_r}\,,\\
    \jp{\vec{v}}_\gamma &\coloneqq
    \vec{v} \vert_{\partial K_l} \cdot \vec{n}_{K_l}\; +\;
    \vec{v} \vert_{\partial K_r} \cdot \vec{n}_{K_r}\,.
  \end{align}
    \end{subequations}
Note that this is consistent with the following definition
  \begin{equation*}
    \jp{u}_\gamma (x) = \left(\lim_{x' \rightarrow x \text{ in } K_l} u(x') -  \lim_{x' \rightarrow x \text{ in } K_r} u(x')\right) \vec{n}_{\gamma}\,.
  \end{equation*}
and correspondingly for the vector-valued function $\vec{v}$.\\
Note further that the jump of a scalar function is vector-valued,
  while the jump of a vector-valued function is scalar.
\end{definition}

\begin{lemma}
  \label{lemma:conservationprop}
  Given a potential $u$ with a flux $\sigma \nabla u$ with continuous normal component along any surface, also the normal component of
  $\sigma \nabla \ucorr + \scorr \nabla \uinf$ is continuous for the subtraction approach.
\end{lemma}

\begin{proof}
  We consider an arbitrary interface $\gamma$. At each point $x$ along
  $\gamma$ the normal components of the fluxes,
  $\sigma \nabla u \cdot \vec{n}_\gamma$ and $\sinf \nabla \uinf \cdot \vec{n}_\gamma$, are continuous. Thus, the
  jump vanishes for them and we obtain
  \begin{equation}
    \jp{\sigma \nabla u}_\gamma = 0 = \jp{\sinf \nabla \uinf}_\gamma.
  \end{equation}
  Rewriting $\jp{\sigma \nabla u}_\gamma$ in terms of
  $\scorr, \sinf, \ucorr$, and $\uinf$, we obtain
  \begin{eqnarray}
    \nonumber
     &
    \jp{\sinf \nabla \ucorr}_\gamma +
    \jp{\scorr \nabla \ucorr}_\gamma +
    \jp{\scorr \nabla \uinf}_\gamma = 0 \\
    \nonumber
    \Leftrightarrow &
    \jp{\sigma \nabla \ucorr}_\gamma =
    - \jp{\scorr \nabla \uinf}_\gamma\\
    \Leftrightarrow &
    \label{eq:flux-jump-relation}
    \jp{\sigma \nabla \ucorr
    + \scorr \nabla \uinf}_\gamma = 0.
  \end{eqnarray}
  As this property holds for any control volume, the normal component of the combined flux $\sigma \nabla \ucorr
    + \scorr \nabla \uinf$ is also continuous for any interface $\gamma$.
\end{proof}

Note that this also implies the identity
\begin{equation}
  \jp{\sigma \nabla \ucorr}_\gamma =
  - \jp{\scorr \nabla \uinf}_\gamma\,,\label{eq:jp_eq}
\end{equation}
for any interface $\gamma$, which is later needed to derive the weak form \eqref{eq:weakform}
from equations \eqref{eq:partint_lhs3} and \eqref{eq:dg_rhs}.

\subsubsection{A weak formulation}
An alternative to the conforming discretization sketched in Section \ref{sec:subtraction}
is to use more general trial and test spaces.
We suggest employing a symmetric discontinuous Galerkin
discretization. The standard derivation of the DG formulation does not apply
immediately, as the intrinsic conservation property for $\ucorr$ differs
from the conservation property of the classical Poisson problem. In
the following section, we will briefly describe the most important steps in
the construction of a \emph{Symmetric Interior Penalty Galerkin} (SIPG) DG formulation for the subtraction approach.
For further details on DG methods, we refer to \cite{Arnold2002} or the book
of Pietro and Ern \cite{AErn1}. We start with the usual definitions:
\begin{definition}[Triangulation $\mathcal{T}_h(\Omega)$]
  Let $\mathcal{T}_h(\Omega)$ be a finite
  collection of disjoint and open subsets forming a partition of
  $\Omega$. The subscript $h$ corresponds to the mesh-width \mbox{$h \coloneqq \max
  \left\{ diam ( E ) \mid E \in \mathcal{T}_h \right\}$}.
Furthermore, the
triangulation induces the \emph{internal
  skeleton}
\begin{equation}
  \label{eq:gamma_int}
  \Gamma_\text{int} \coloneqq \left\{ \gamma_{e,f} = \partial E_e
    \cap \partial
    E_f \; \right\vert
  \left. E_e, E_f \in \mathcal{T}_h \, ,\, E_e \neq E_f,\, \left|
      \gamma_{e,f} \right| > 0 \right\}
\end{equation}
and the \emph{skeleton} $\Gamma \coloneqq \Gamma_\text{int} \cup \partial\Omega$.
\end{definition}

\begin{definition}[Broken polynomial spaces]  Broken polynomial spaces are defined as piecewise polynomial spaces on the partition
  $\mathcal{T}_h(\Omega)$ as
\begin{equation}
    \label{eq:brokenspaces}
    V^k_h \coloneqq \left\{\, v \in L^2(\Omega)
      \,:\, v|_{E} \in P^k(E) \,\right\},
\end{equation}
  where $P^k$ denotes the space of polynomial functions of degree $k$.
  They describe functions that exhibit elementwise polynomial behavior
  but may be discontinuous across element interfaces.
\end{definition}

Since the elements of $V^k_h$ may admit discontinuities across element boundaries, the gradient of a function $v \in V^k_h$ is not defined everywhere on $\Omega$. To account for this, we introduce the broken gradient operator.

\begin{definition}[Broken gradient operator] The broken gradient $\nabla_h : V^k_h \rightarrow [L^k(\Omega)]^d$ is defined such that, for all $v \in V^k_h$
\begin{equation}
	\left(\nabla_h v\right)|_E = \nabla(v|_E) \qquad \text{ for all } E \in \mathcal{T}_h(\Omega).
\end{equation}
\end{definition}
\begin{notdef}[Jump and Average]
Using the Definition \ref{def:jump} we introduce the abbreviated notation
of the jump
\begin{equation*}
\jp{x}_{e,f} \coloneqq \jp{x}_{\gamma_{e,f}}
\end{equation*}
of a piecewise continuous function $x$ on the interface $\gamma_{e,f}$
between two adjacent elements $E_e, E_f \in \mathcal{T}_h$.
We further define the \emph{average operator}
\begin{align}
\av{x}_{e,f} \coloneqq \omega_{e,f} x\vert_{\partial E_e} + \omega_{f,e} x\vert_{\partial E_f}\,.
\end{align}

The weights
$\omega_{e,f}$ and $\omega_{f,e}$ can be chosen to be the arithmetic
mean, but for the case of heterogeneous conductivities, \mbox{\cite{dipietro2008}} has
shown  that a conductivity-dependent choice is optimal:
\begin{align}
\omega_{e,f} \coloneqq \frac{\sigma_f}{\sigma_f+\sigma_e} \text{\quad{}and\quad} \omega_{f,e} \coloneqq \frac{\sigma_e}{\sigma_e+\sigma_f}\,.
\end{align}
We further introduce the average operator with switched weights
\begin{align}
\av{x}^{\ast}_{e,f} := \omega_{f,e} x\vert_{\partial E_e} + \omega_{e,f} x\vert_{\partial E_f},
\end{align}
 and obtain the following multiplicative property:
\begin{align}
  \label{eq:jumpprop}
  \jp{x y}_{e,f} =
  \jp{x}_{e,f} \av{y}^{\ast}_{e,f} + \av{x}_{e,f} \jp{y}_{e,f}\,.
\end{align}
\end{notdef}

Using a Galerkin approach, we seek for a solution $\ucorr_h \in V^k_h$,
which fulfills \eqref{eq:subtraction-strong} in a weak sense. We
start the derivation by testing with a test function $v_h \in V^k_h$:
\begin{equation}
  \label{eq:dg-partint}
  - \int\limits_{\Omega} \nabla \cdot \sigma \nabla \ucorr_h  v_h \, dx =
  \int\limits_{\Omega} \nabla \cdot \scorr \nabla \uinf  v_h \, dx
\end{equation}

On each $E \in \mathcal{T}_h(\Omega)$, we apply integration by
parts. Element boundaries are split into the domain
boundary and all internal edges. The electrical current $\sigma \nabla \ucorr_h \cdot \vec{n}$ through the
boundary is given by the inhomogeneous Neumann boundary
conditions \eqref{eq:subtraction-strong2}.
For the left-hand side, we obtain
\begin{equation}
  \label{eq:partint_lhs}
\begin{split}
  \text{lhs}
  =&- \int\limits_{\Omega} \nabla \cdot \sigma \nabla \ucorr_h  v_h \, dx\\
  =&  \int\limits_{\Omega} \sigma
    \nabla_h \ucorr_h \cdot \nabla_h v_h \, dx
    + \int\limits_{\partial \Omega}
      \sigma \nabla \uinf \cdot \vec{n}\, v_h \, ds
    - \int\limits_{\Gamma_\text{int}}
    \jp{\sigma \nabla_h \ucorr_h v_h} \, ds,
\end{split}
\end{equation}

and with the multiplicative property \eqref{eq:jumpprop} follows
\begin{equation}
\begin{split}
  \label{eq:partint_lhs3}
  \text{lhs} =
  \int\limits_{\Omega} \sigma
  \nabla_h \ucorr_h \cdot \nabla_h v_h \, dx
  &+ \int\limits_{\partial \Omega}
  \underbrace{\sigma \nabla \uinf \cdot \vec{n}\, v_h}_{\text{term}~\dagger} \, ds \\
  &- \int\limits_{\Gamma_\text{int}}
  \underbrace{\jp{\sigma \nabla_h \ucorr_h}}_{\text{term}~\ddag} \av{v_h}^{\ast}  +
  \av{\sigma \nabla_h \ucorr_h} \jp{v_h}\, ds\, .
\end{split}
\end{equation}
Applying the same relations for the right-hand side, we obtain
\begin{equation}
\begin{split}
  \text{rhs} =
  - \int\limits_{\Omega} \scorr \nabla \uinf \cdot
  \nabla_h v_h \,dx
  &+ \int\limits_{\partial \Omega}
  \underbrace{\scorr \nabla \uinf \cdot \vec{n}\, v_h}_{\text{term}~\dagger} \,ds \\
  &+ \int\limits_{\Gamma_\text{int}}
  \underbrace{\jp{\scorr \nabla \uinf}}_{\text{term}~\ddag}\av{v_h}^{\ast}
  + \av{\scorr\nabla \uinf} \jp{v_h}  \,ds \, . \label{eq:dg_rhs}
\end{split}
\end{equation}

Summing up the boundary integrals \eqref{eq:partint_lhs3}$^\dagger$ and
\eqref{eq:dg_rhs}$^\dagger$ yields a remaining term $-\sinf \nabla
\uinf \cdot \vec{n}\, v_h$ on the right-hand side.
As discussed in Section \ref{sec:conservation}, the conservation
properties also imply that the normal component of $\sigma \nabla \ucorr + \scorr \nabla \uinf$ is
continuous, see \eqref{eq:flux-jump-relation}. For the discrete
solution, we require the same conservation property; thus the jump term
\eqref{eq:partint_lhs3}$^\ddag$ equals to
$-\jp{\scorr \nabla \uinf}$ and cancels out with term
\eqref{eq:dg_rhs}$^\ddag$.

To gain adjoint consistency, we symmetrize the operator and add the additional term
\begin{equation} \label{consistencyterm}
  \tilde{a}^{_\text{sym}}(\ucorr_h, v_h) := - \int\limits_{\Gamma_\text{int}} \av{\sigma \nabla_h v_h} \jp{\ucorr_h}  \,ds\,.
\end{equation}
To guarantee coercivity,
 the left-hand side is
supplemented with
the
penalty term
\begin{align}
  J(\ucorr_h, v_h) =
  \eta
  \int\limits_{\Gamma_\text{int}}
  \frac{\hat\sigma_\gamma}{h_\gamma}
  \jp{ \ucorr_h } \jp{ v_h } \,ds\,,
  \label{penaltyterm}
\end{align}
where $h_\gamma$ and $\hat\sigma_\gamma$ denote local definitions of
the mesh width and the electric conductivity on an edge
$\gamma$, respectively.
In our particular case, we choose $h_\gamma$ according to \cite{houston2012anisotropic} and
$\hat\sigma_\gamma$ as the harmonic average of the
conductivities of the adjacent elements \mbox{\cite{dipietro2008}}:
\begin{align*}
 h_{\gamma_{e,f}}=\frac{\min(|E_e|,|E_f|)}{|\gamma_{e,f}|} \text{\quad{}and\quad} \hat\sigma_{\gamma_{e,f}} := \frac{2\sigma_e\sigma_f}{\sigma_e+\sigma_f}\,.
\end{align*}
The penalty parameter
$\eta$ has to be chosen large enough to ensure
coercivity.

This derivation yields the SIPG
formulation \cite{wheeler1978,riviere2002pee}
or for weighted averages the \emph{Symmetric Weighted
  Interior Penalty Galerkin} (SWIPG or SWIP) method\mbox{\cite{dipietro2008}}:

Find $\ucorr_h \in V_h$ such that
\begin{subequations}
  \label{eq:subtraction-sipg}
  \begin{align}
    &a(\ucorr_h,v_h) + J (\ucorr_h,v_h) = l(v_h) \quad \text{ for all } v_h\in V_h, \label{eq:weakform}\\
    \intertext{with}
    \nonumber
    a(\ucorr_h, v_h) = & ~\tilde{a}(\ucorr_h, v_h) + \tilde{a}^{_\text{sym}}(\ucorr_h, v_h) \\
      = &
      \int\limits_{\Omega} \sigma \nabla_h \ucorr_h \cdot \nabla_h v_h \,dx
      - \int\limits_{\Gamma_\text{int}} \av{ \sigma \nabla_h \ucorr_h } \jp{v_h} +
       \av{\sigma \nabla_h v_h} \jp{\ucorr_h} \,ds\,,\\
    J(\ucorr_h, v_h) =&~ \eta \int\limits_{\Gamma_\text{int}}
       \frac{\hat\sigma_\gamma}{h_\gamma}
       \jp{ \ucorr_h } \jp{ v_h } \,ds\,,\\
    \nonumber
    l(v_h) = & - \int\limits_{\Omega} \scorr \nabla \uinf \cdot \nabla_h v_h \,dx\\
      & + \int\limits_{\Gamma_\text{int}} \av{ \scorr \nabla \uinf } \jp{v_h} \,ds
        - \int\limits_{\partial \Omega} \sinf \partial_{\mathbf{n}} \uinf v_h \,ds\,.
  \end{align}
\end{subequations}
Given the correction potential $\ucorr_h$, the full potential $u_h$ can be
reconstructed as $u_h = \ucorr_h + \uinf$.

\begin{rem}[Discrete Properties]
  As $a(\ucorr_h,v_h)$ and $J(\ucorr_h,v_h)$ are the same operators as
  in \mbox{\cite{dipietro2008}}, the following properties follow immediately: The
  proposed SIPG discretization \eqref{eq:subtraction-sipg} is
  consistent and adjoint-consistent with the strong problem
  \eqref{eq:subtraction-strong}, and for a sufficiently large constant
  $\eta > 0$ it has a unique solution.
\end{rem}
\begin{rem}[Conservation Property]
  Furthermore, for $K \subseteq \mathcal{T}_h(\Omega)$, \eqref{eq:subtraction-sipg} fulfills a discrete
  conservation property
  \begin{equation*}
    \int\limits_{\mathclap{\partial K}}
    \underbrace{\av{\sigma \nabla_h \ucorr_h}
      - \eta \frac{\hat{\sigma}_\gamma}{\hat{h}_\gamma} \jp{\ucorr_h}}_{\vec{j}_h^\corr}\,ds =
    \int\limits_{K}
    \underbrace{- \nabla \scorr \nabla \uinf}_{f^\corr} \,ds
  \end{equation*}
  with the discrete flux
  $\vec{j}^\corr_h$.
  For $h\rightarrow 0$, the jump $\jp{\ucorr_h}$ vanishes and the
  discrete flux $\vec{j}^\corr_h$ converges to the flux
  $\vec{j}^\corr$ as defined in \eqref{eq:continuity2}.
\end{rem}
\section{Methods}

\subsection{Implementation and parameter settings}
We implemented the DG-FEM subtraction approach in the
DUNE framework \cite{dune08:1,dune08:2} using the DUNE PDELab toolbox
\cite{bastian2010generic}. For reasons of comparison, we also implemented the
CG-FEM subtraction approach in the same framework.
We use linear Ansatz-functions
 for both the DG (i.e., $k=1$ in \eqref{eq:brokenspaces}) and CG
approaches throughout this study. On a given triangulation
$\mathcal{T}_h$, we choose basis
functions $\{\phi^i_h\}$, $i \in [0,N_h)$, with local support, where
$N_h$ denotes the number of unknowns. The penalty parameter $\eta$ was chosen to $\eta$ = 0.39.
For the CG simulations, a Lagrange basis with the usual hat functions
is employed, whereas for the DG case, elementwise $L^2$-orthonormal functions are
chosen. In this setup ($k=1$, hexahedral mesh), we
have eight unknowns per mesh cell for the DG approach, i.e., $N_h = 8
\times \#\text{cells}$, and one unknown per vertex for the CG
approach, i.e., $N_h = \#\text{vertices}$.
Evaluating the bilinear forms $a(\cdot,\cdot)$, $J(\cdot,\cdot)$, and
the right-hand side $l(\cdot)$ leads to a linear system $A \cdot x =
b$, where $x \in \mathbb{R}^{N_h}$
denotes the coefficient vector, and the approximated solution of \eqref{eq:subtraction-sipg} is $u_h^{corr} = \sum_i x_i \phi^i_h$.
Furthermore, $A \in \mathbb{R}^{N_h \times N_h}$
is the matrix representation of the bilinear
operator $a+J$ and $b \in \mathbb{R}^{N_h}$ the right-hand-side vector:
\begin{align*}
  A_{ij} &= a(\phi^j_h,\phi^i_h) + J(\phi^j_h,\phi^i_h) & i,j & \in [0,N_h),\\
  b_{i} &= l(\phi^i_h) & i & \in [0,N_h)\,.
\end{align*}
The resulting matrix $A$ has a sparse block structure with small dense
blocks, in our case of dimension $8 \times 8$. The outer
structure is similar to that of a finite volume discretization,
i.e., rows corresponding to each grid cell and one off-diagonal entry
for each cell neighbor.
By now, a range of efficient solvers for DG discretizations
is available, using multigrid \cite{NLA:NLA1816} or
domain decomposition methods \cite{Antonietti:2011:DG-DD}. The computation/solving times for the CG- and DG-FEM for realistic six-layer head models and a realistic EEG sensor configuration are compared in the results section.

\subsection{Volume conductor models}
To validate and compare the accuracy of these numerical schemes, we used
four-layer sphere volume conductor models, where an analytical solution exists and can
be used as a reference \cite{CHW:Mun93}.
\begin{table}[!t]
\renewcommand{\arraystretch}{1.3}
\caption{Conductive compartments (from in- to outside)}
\label{tab:compartments}
\centering
\begin{tabular}{lrr}
\hline
Compartment & \multicolumn{1}{l}{Outer Radius} & \multicolumn{1}{l}{Conductivity} \\
\hline \hline
Brain & 78 mm & 0.33 S/m \\
CSF & 80 mm & 1.79 S/m \\
Skull & 86 mm & 0.01 S/m \\
Skin & 92 mm & 0.43 S/m \\
\hline
\end{tabular}
\end{table}
For the four spherical compartments, representing
brain, cerebrospinal fluid (CSF), skull, and skin, we chose radii and conductivities as shown in
Table \ref{tab:compartments}. As discussed in the introduction and in \ref{sec:skullleakage},
we used hexahedral meshes in our study. To be able to distinguish between
numerical and geometry errors, i.e., errors due to the discrete approximation of the continuous PDE and errors due to an inaccurate representation of the geometry, respectively, we constructed a variety of head
models with different segmentation resolutions (1 mm, 2 mm, and
4 mm) and for each of these we again used different mesh resolutions
(1 mm, 2 mm, and 4 mm).
\begin{figure*}[!t]%
\centering
\setlength{\tabcolsep}{.005\textwidth}
\begin{tabular}{ccc}
\segres{1}{1} & \segres{2}{2} & \segres{4}{4} \\
\includegraphics[width=.3\textwidth]{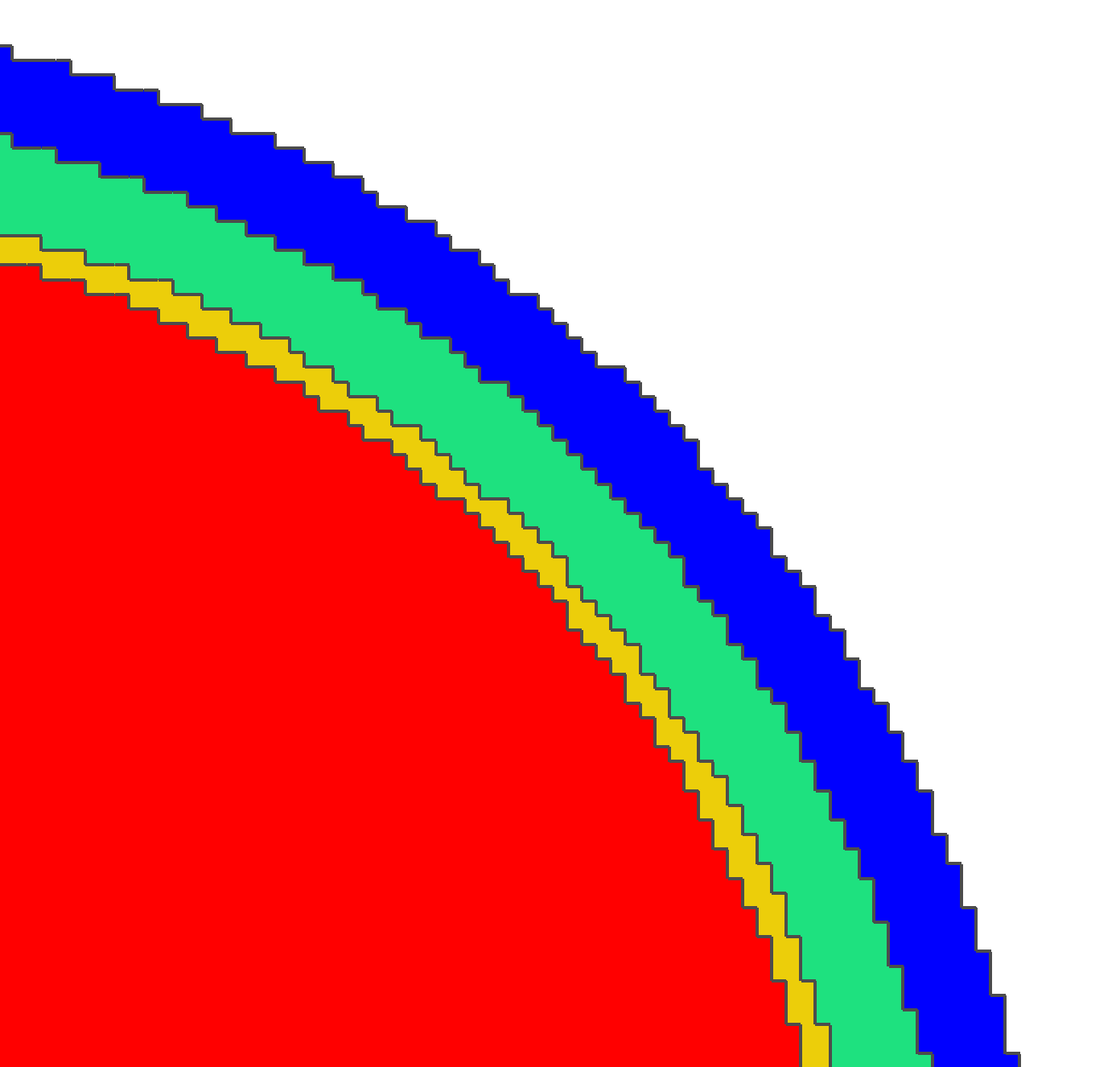} & \includegraphics[width=.3\textwidth]{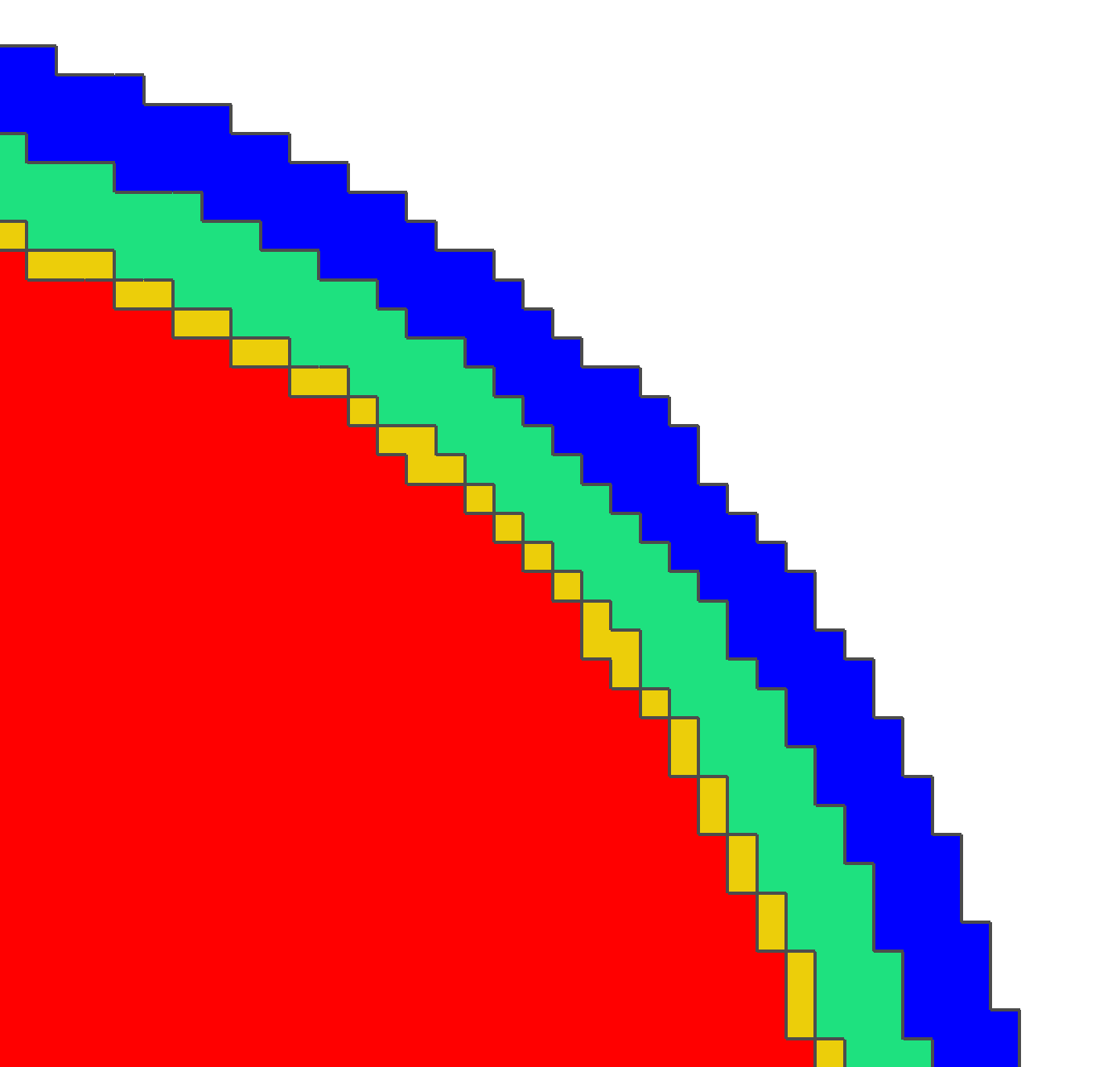} & \includegraphics[width=.3\textwidth]{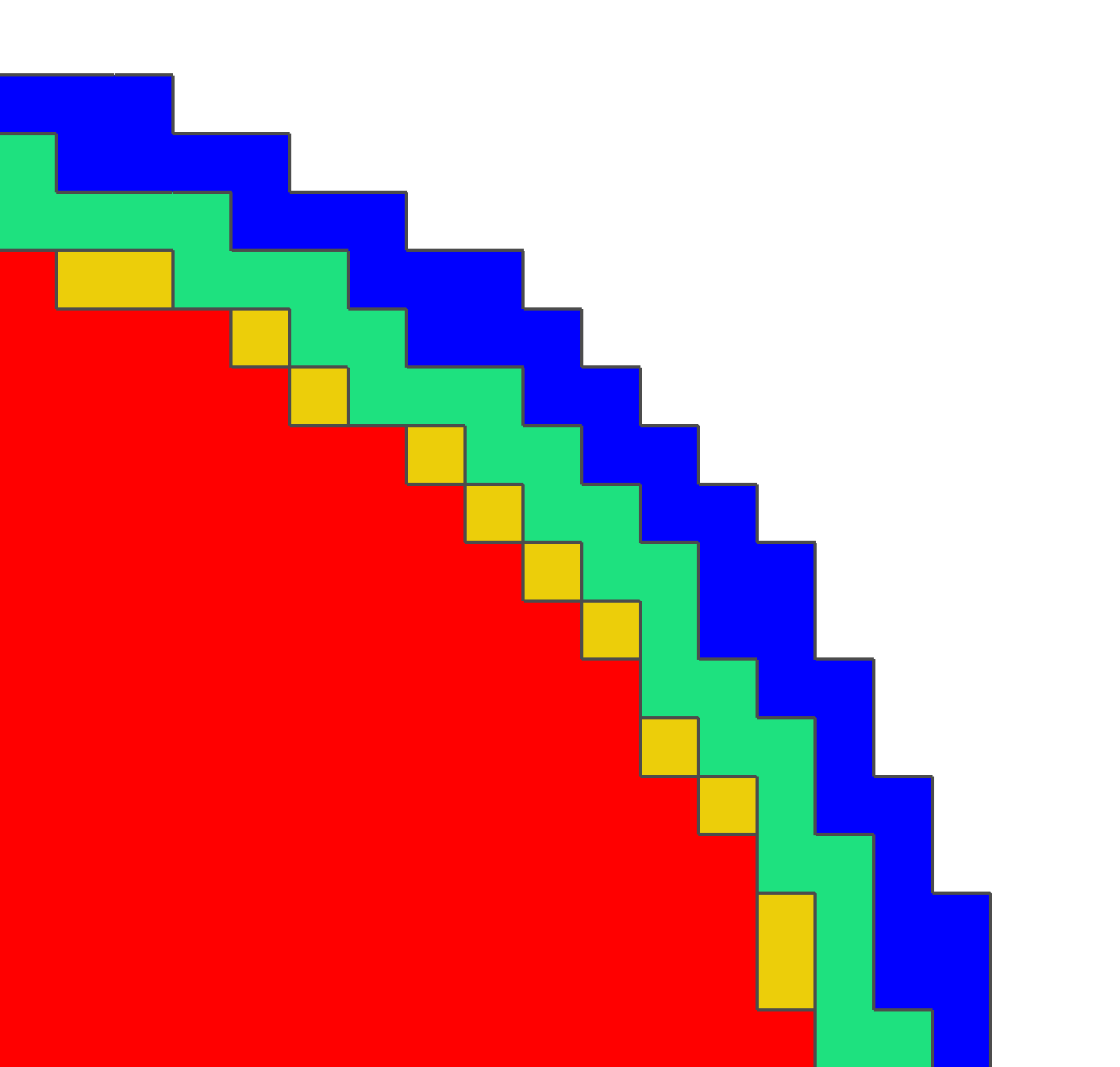}
\end{tabular}
\caption{Visualization of models \segres{1}{1}, \segres{2}{2} and \segres{4}{4} (from left to right), cut in x-plane at the origin; coloring is brain - red, CSF - yellow, skull - green, skin - blue.}%
\label{fig:vis-mesh}%
\end{figure*}
\begin{table}[!t]
\renewcommand{\arraystretch}{1.3}
\caption{Model properties (from left to right): segmentation resolution (\it{Seg.}), mesh resolution (\it{h}), number of nodes, number of elements}
\label{tab:models}
\centering
\begin{tabular}{lrrrr}
\hline
& Seg. & \multicolumn{1}{l}{h} & \multicolumn{1}{l}{\#nodes} & \multicolumn{1}{l}{\#elements} \\
\hline \hline
\segres{1}{1} & 1 mm & 1 mm & 3,342,701 & 3,262,312 \\
\segres{2}{1} & 2 mm & 1 mm & 3,343,801 & 3,263,232 \\
\segres{2}{2} & 2 mm & 2 mm & 428,185 & 407,907 \\
\segres{4}{1} & 4 mm & 1 mm & 3,351,081 & 3,270,656 \\
\segres{4}{2} & 4 mm & 2 mm & 429,077 & 408,832 \\
\segres{4}{4} & 4 mm & 4 mm & 56,235 & 51,104 \\
\hline
\ci{6}{hex}{1mm} & 1 mm & 1 mm & 3,965,968 & 3,871,029 \\
\ci{6}{hex}{2mm} & 2 mm & 2 mm & 508,412 & 484,532 \\
\ci{6}{tet}{hr} & - & - & 2,242,186 & 14,223,508 \\
\hline
\end{tabular}
\end{table}
The details of these head models are listed in Table \ref{tab:models},
and Figure \ref{fig:vis-mesh} visualizes a subset of the used models.

To further evaluate the sensitivity of the different numerical
methods to leakage effects, we intentionally generated spherical
models with skull leakages.
Therefore, we chose the model \segres{2}{2} and reduced the
radius of the outer skull boundary to 82 mm, 83 mm, and 84 mm,
resulting in skull thicknesses of 2 mm, 3 mm, and 4 mm, respectively.
This way, we were able to generate a leakage scenario similar to the
one presented in Figure \ref{fig:shortcut}, while preserving the
advantage of a spherical solution that can be used for error
evaluations.
\begin{table}[!t]
\renewcommand{\arraystretch}{1.3}
\caption{Model parameters}
\label{tab:leaks}
\centering
\begin{tabular}{lcrrr}
\hline
& \multicolumn{1}{l}{out. Skull Rad.} & \multicolumn{1}{l}{\#leaks} \\
\hline \hline
\segresR{2}{2}{82} & 82 mm & 10,080 \\
\segresR{2}{2}{83} & 83 mm & 1,344 \\
\segresR{2}{2}{84} & 84 mm & 0 \\
\hline
\end{tabular}
\end{table}
Table \ref{tab:leaks} indicates the number of leaks for each model, i.e.,  the
number of vertices belonging to both an element labeled as skin and an
element labeled as CSF or brain.

\subsection{Sources}
Since the numerical accuracy depends on the local mesh structure and the source eccentricity, we used 10 source eccentricities and, for each eccentricity,
randomly distributed 10 sources. Thereby, the variability of the numerical accuracy can be captured for each eccentricity. We evaluated the accuracy for both radial and
tangential dipole directions; however, we present the results only for
radial directions here. The results for dipoles with tangential
direction are very similar to these with slightly lower errors than for radial dipoles.

To make the effect of skull leakage more accessible, we additionally
generated visualizations of the current for one dipole fixed at
position $(1,47,47)$, which corresponds to an element center, and
fixed direction $(0,1,1)$ for both the CG- and DG-FEM and for
all three models with reduced skull thickness as shown in Table \ref{tab:leaks}.
We visualized a cut
through the x-plane at the dipole position and chose to visualize both
the direction and strength of the electric flux for each numerical
method and model (Figure \ref{fig:vis-current}). Furthermore, the relative
change in strength and the flux difference between the
numerical methods, described by the metrics \lnMAGj and \totDIFFj as defined in the next section, were visualized for each model (Figure \ref{fig:vis-leak}).

\subsection{Error metrics}
To achieve a result that purely represents the numerical
and segmentation accuracy and is independent of the chosen sensor configuration,
we evaluated the solutions on the whole outer layer. We use two error measures to
distinguish between topography and magnitude errors,
the relative difference measure (RDM),
\begin{equation}
\label{eq:eval:rdm}
\begin{split}
  RDM ( u_h, u ) &= \left\| \frac{u_h}{\| u_h \|_2} - \frac{u}{\| u \|_2}  \right\|_2,
 \end{split}
\end{equation}
and the logarithmic magnitude error (lnMAG),
\begin{equation}
\label{eq:eval:mag}
\begin{split}
 lnMAG ( u_h, u ) &= \ln \left(\frac{\| u_h \|_2}{\| u \|_2}\right).
 \end{split}
\end{equation}
Besides presenting the mean RDM and lnMAG errors over all sources at a certain eccentricity (see, e.g., left subfigures in Figure \ref{fig:radial-geom}), we also present results in separate boxplots (see, e.g., right subfigures in Figure \ref{fig:radial-geom}). The boxplots
show maximum and minimum error over all source positions at a certain eccentricity, indicated by upper and lower error bars. This allows to display the overall variability of the error. Furthermore, the boxplots show the upper and lower quartiles. The interquartile range is marked by a box; a black dash shows the median. Henceforth, the interquartile range will also be denoted as spread.
 Note the different presentation of source eccentricity on the x-axes in
the left and right subfigures.

To evaluate the local changes of the current, we furthermore visualize for each mesh element $E$ the logarithm of the local
change in current magnitude
\begin{equation}
  \lnMAGj(E) =  \ln \left(\frac{\|\vec{j}_{h,\text{\tiny CG}}(x_E)\|_2}{\|\vec{j}_{h,\text{\tiny DG}}(x_E)\|_2}\right),
\end{equation}
and the total local current difference
\begin{equation}
  \totDIFFj(E) = \vec{j}_{h,\text{\tiny CG}}(x_E) - \vec{j}_{h,\text{\tiny DG}}(x_E),
\end{equation}
where $x_E$ denotes the centroid of mesh element $E$ (see Figure \ref{fig:vis-leak}).

We can exploit that, due to the relation $\ln(1 + x) \approx x$ for small $\|x\|$, we have $lnMAG \approx \|u_h\|_2 / \|u\|_2 - 1$ for small deviations. In consequence, $100 \cdot lnMAG$ is about the change of the magnitude in percent. The same approximations are valid for the \lnMAGj.

\subsection{Realistic head model}
\begin{figure}[tb]%
\centering
\includegraphics[width=.4\textwidth]{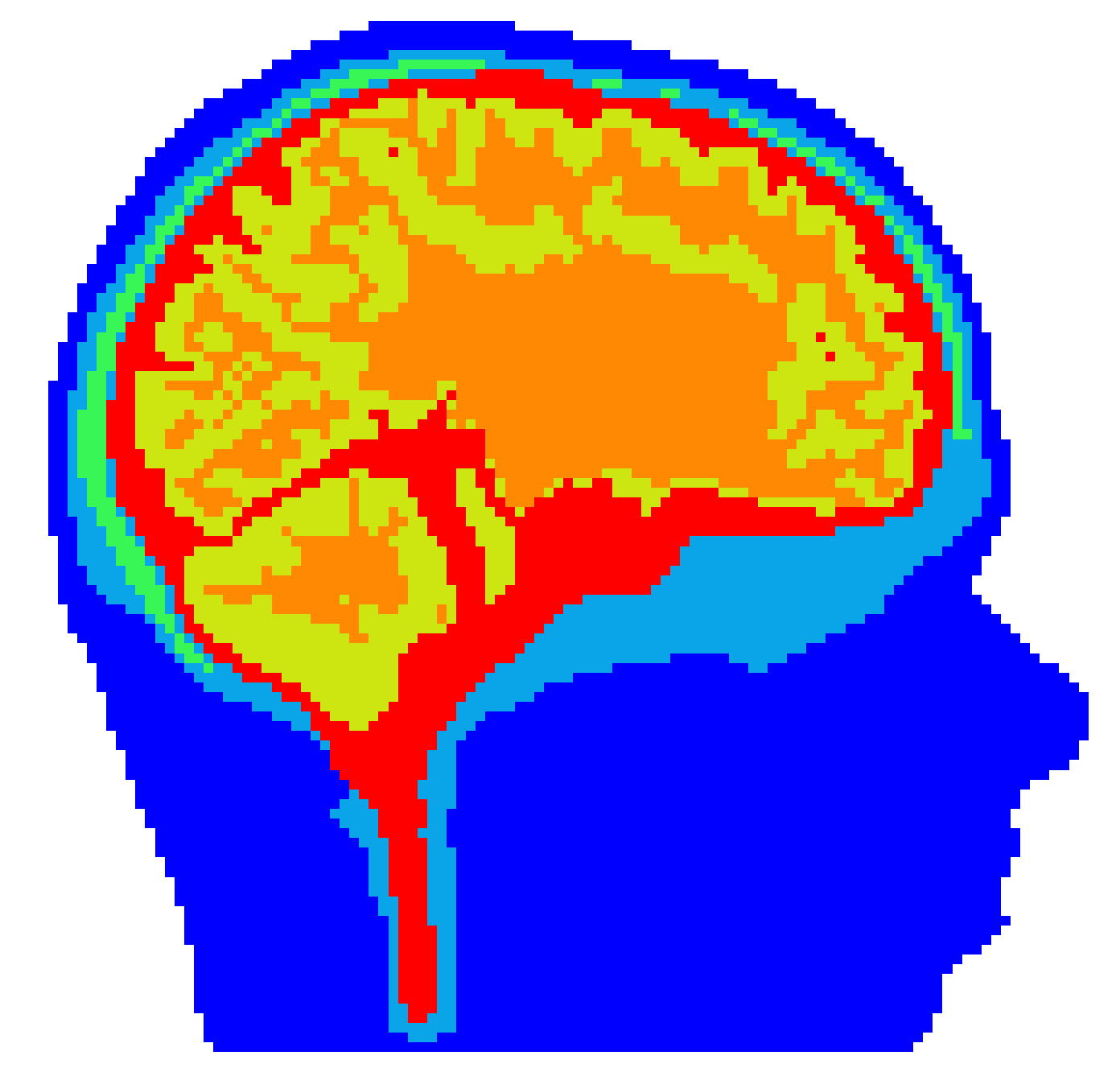} \hfil
\includegraphics[width=.4\textwidth]{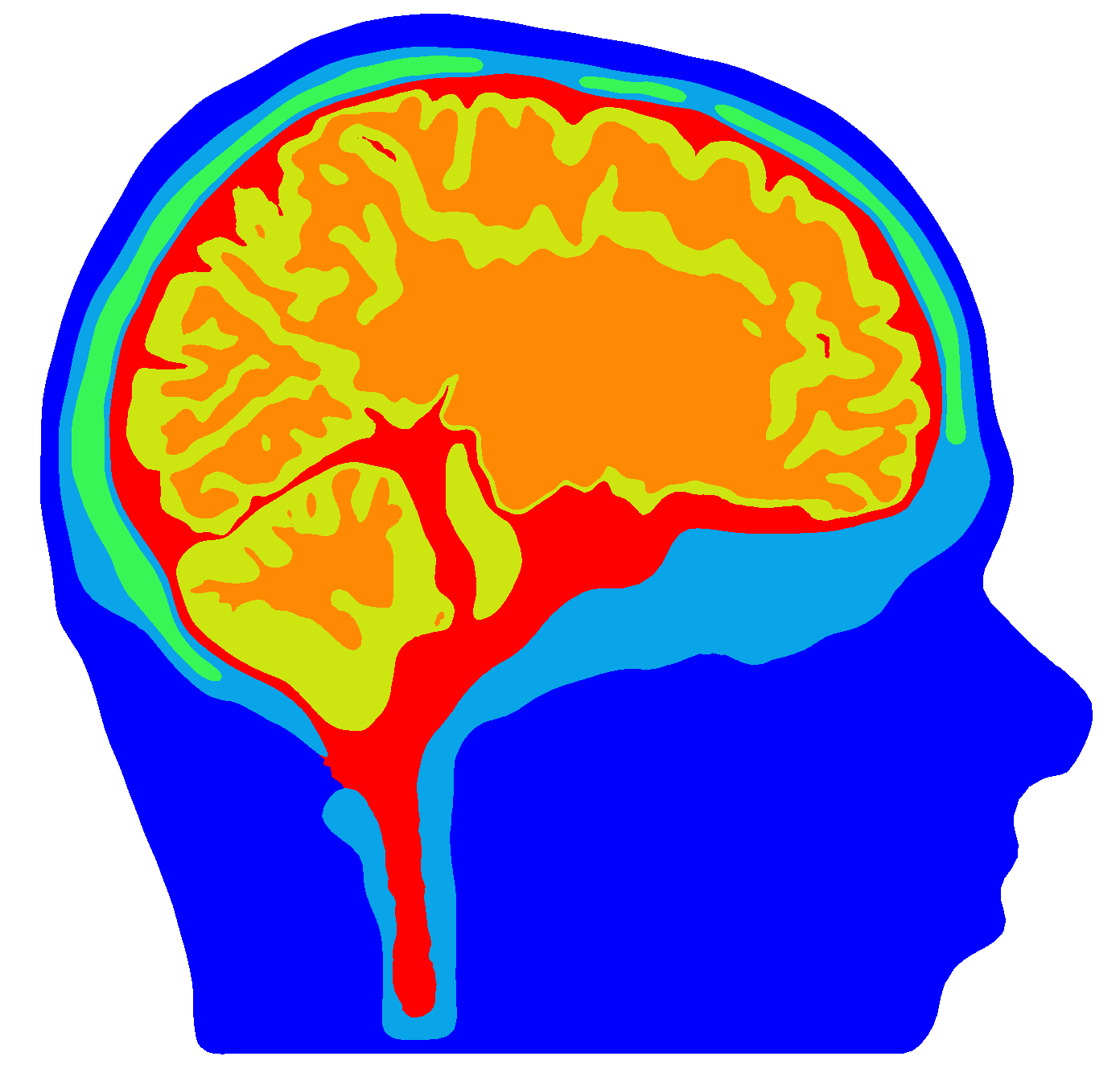}
\caption{Visualization of realistic six-compartment hexahedral \ci{6}{hex}{2mm}, $h = 2 mm$, (left) and high-resolution reference head model \ci{6}{tet}{hr} (right).}%
\label{fig:headmodel}%
\end{figure}

To complete the numerical evaluations, the differences between the CG- and DG-FEM were evaluated in a more realistic scenario. Based on MRI recordings, a segmentation considering six tissue compartments (white matter, gray matter, cerebrospinal fluid, skull compacta, skull spongiosa, and skin) that includes realistic skull openings such as the foramen magnum and the optic nerve canal was generated. Based on this segmentation, three realistic head models were generated. Two hexahedral head models with mesh resolutions of 1 mm and 2 mm, \ci{6}{hex}{1mm} and \ci{6}{hex}{2mm}, were generated, resulting in 3,965,968 vertices and 3,871,029 elements, and 508,412 vertices and 484,532 elements, respectively (Figure \mbox{\ref{fig:headmodel}}). For both models, the segmentation resolution is identical to the mesh resolution. As the model with a mesh width of 2 mm was not corrected for leakages, 1,164 vertices belonging to both CSF and skin elements were found. These leakages were mainly located at the temporal bone. To calculate reference solutions, a high-resolution tetrahedral head model with 2,242,186 vertices and 14,223,508 elements, \ci{6}{tet}{hr}, was generated. For further details of this model and of the used segmentation, please refer to \mbox{\cite{JVorw2014,VorwerkDiss2016}}. The conductivities were chosen according to \mbox{\cite{JVorw2014}}. 4,724 source positions were placed in the gray matter with a normal constraint, and those that were not fully contained in the gray matter compartment, i.e., where the source was placed in an element at a compartment boundary, were excluded. As a result, 4,482 source positions remained for the 1 mm model and 4,430 source positions for the 2 mm model. An 80 channel realistic EEG cap was chosen as the sensor configuration. For both the CG- and DG-FEM, solutions in the 1 mm and 2 mm hexahedral head model were computed and the RDM and lnMAG are evaluated in comparison to the solution of the CG-FEM calculated using the tetrahedral head model.

The computations were performed on a Linux-PC with an Intel Xeon E5-2698 v3 CPU (2.30 GHz). The computation times for the CG- and DG-FEM in the models \ci{6}{hex}{1mm} and \ci{6}{hex}{2mm} were evaluated in the results section. Though an optimal speedup through parallelization can be achieved for both the transfer matrix computation and the right-hand-side setup, all computations were carried out without parallelization on a single core to allow for a reliable comparison.

\section{Results}
\begin{figure*}[!t]%
\centering
\includegraphics[width=.48\textwidth]{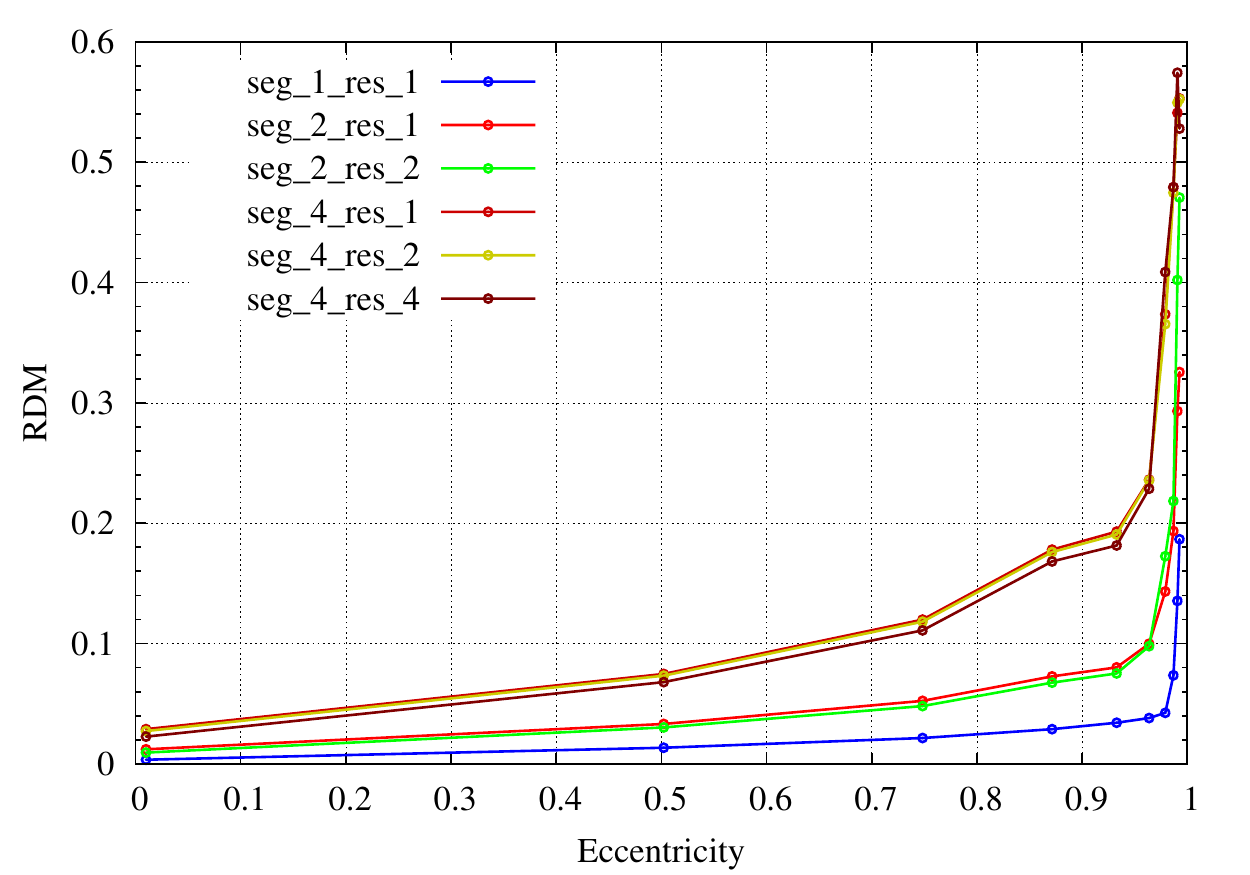} \hfill \includegraphics[width=.48\textwidth]{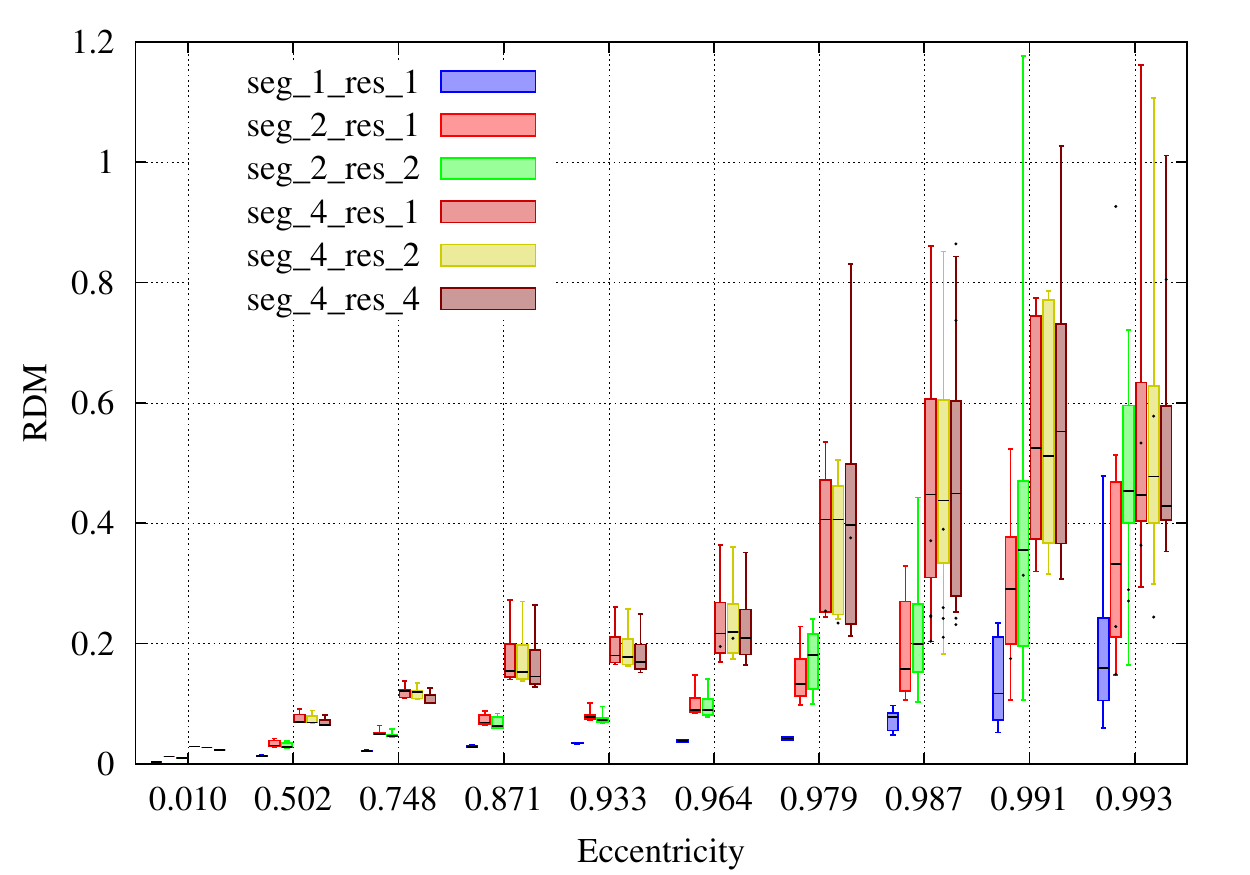}\\%
\includegraphics[width=.48\textwidth]{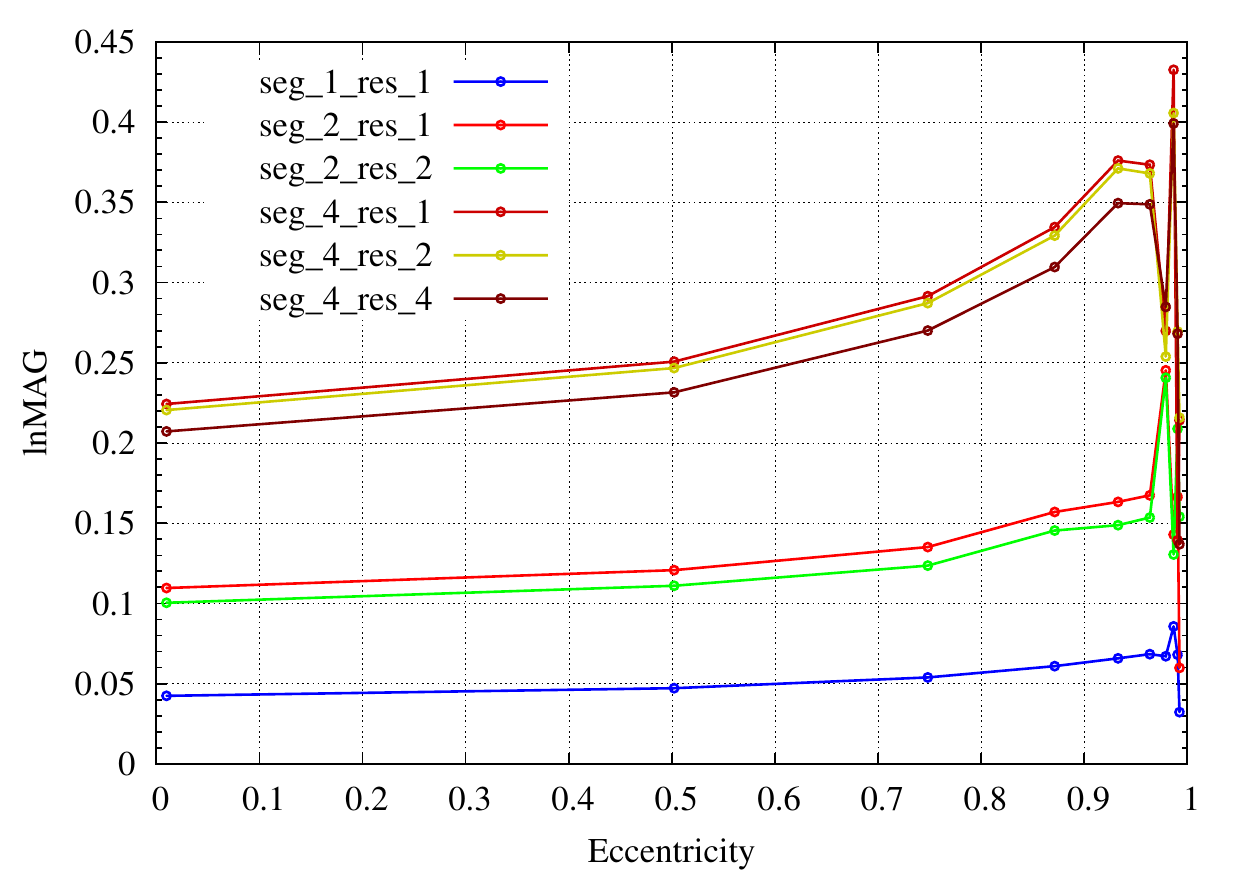} \hfill \includegraphics[width=.48\textwidth]{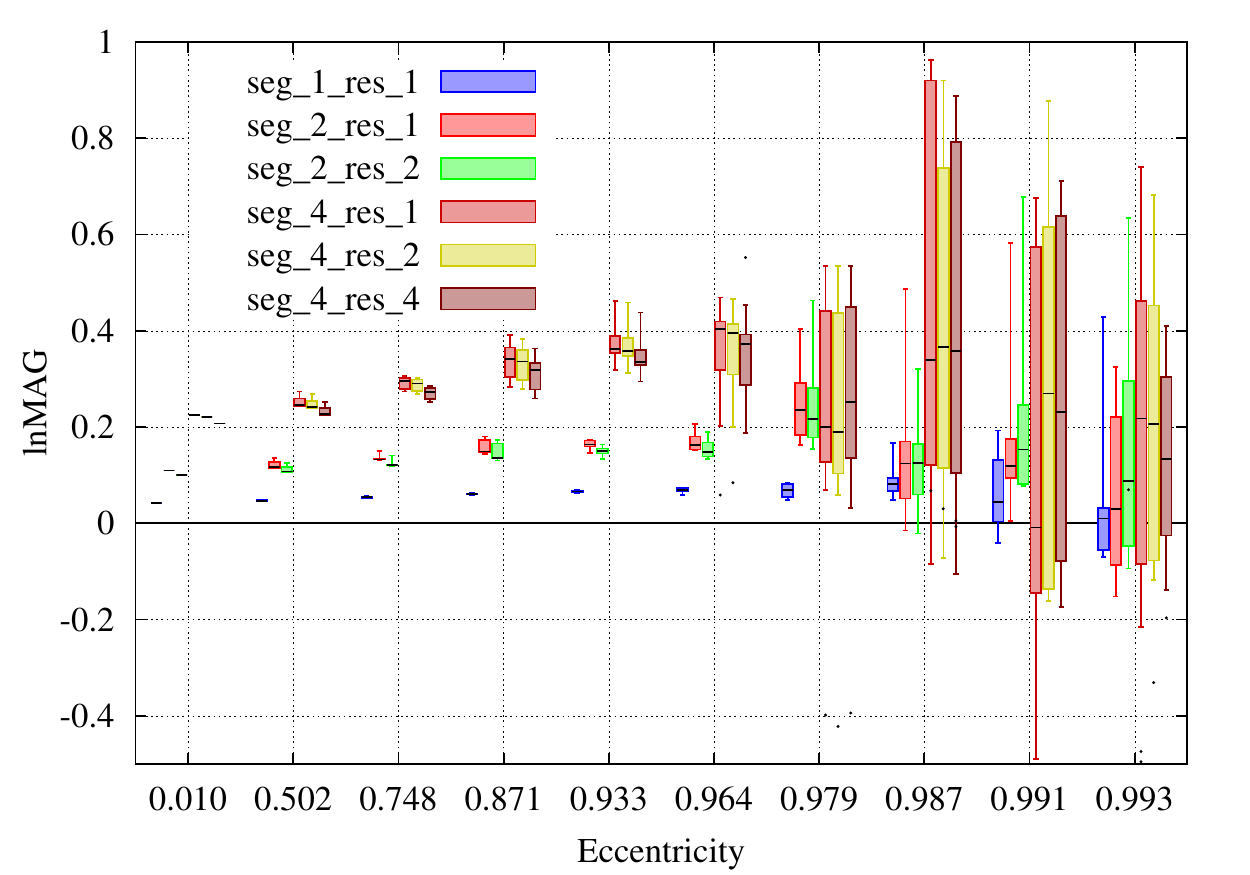}
\caption{Convergence for the DG-FEM with increasing mesh and/or segmentation resolution. Results of radial dipole computations. Visualized are the mean error (left column) and boxplots (right column) of the RDM (top row) and lnMAG (bottom row). Dipole positions that are outside the brain compartment in the discretized models are marked as dots.
Note the different scaling of the x-axes. Note that the error curve for model \segres{4}{1} is partly covered by that of model \segres{4}{2} in the top left figure.}%
\label{fig:radial-geom}%
\end{figure*}

\begin{figure*}[!t]%
\centering
\includegraphics[width=.48\textwidth]{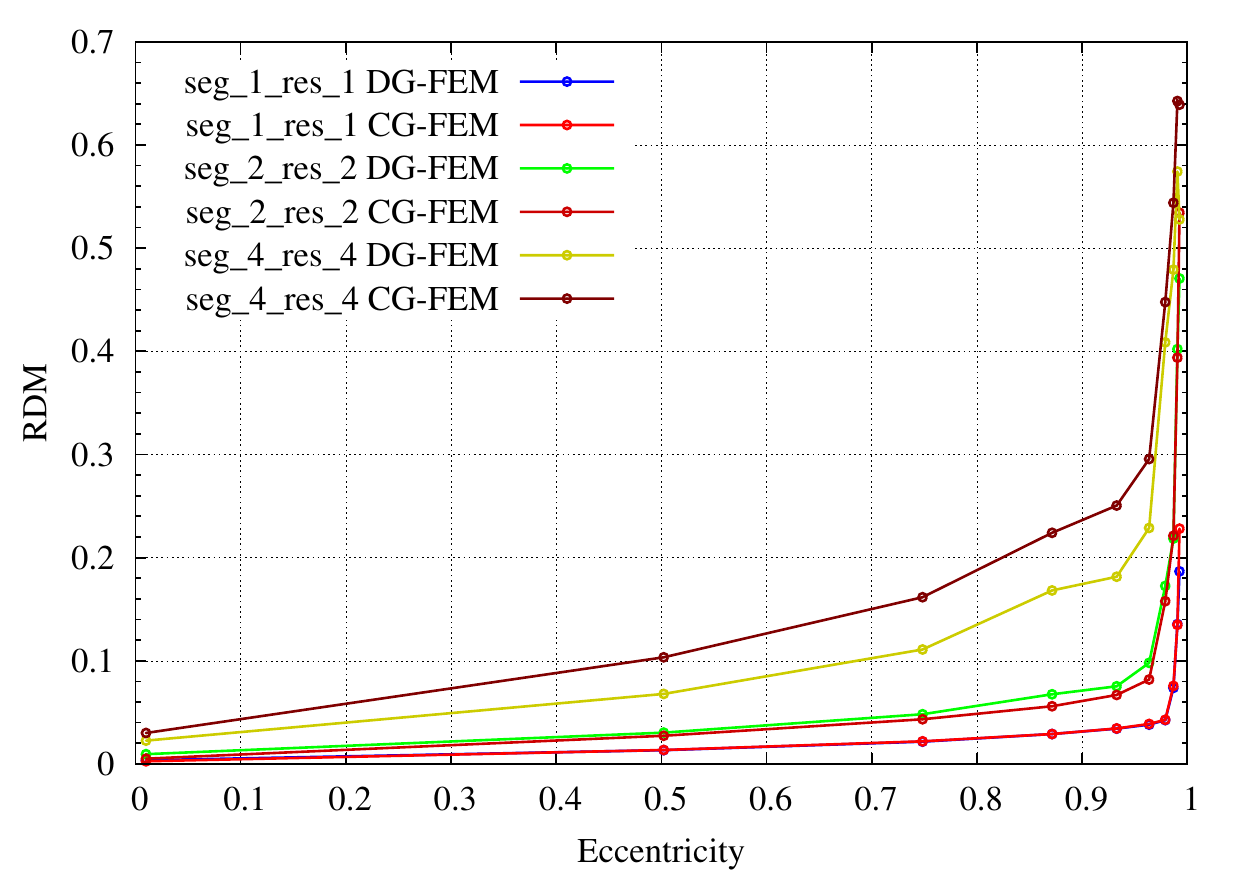} \hfill \includegraphics[width=.48\textwidth]{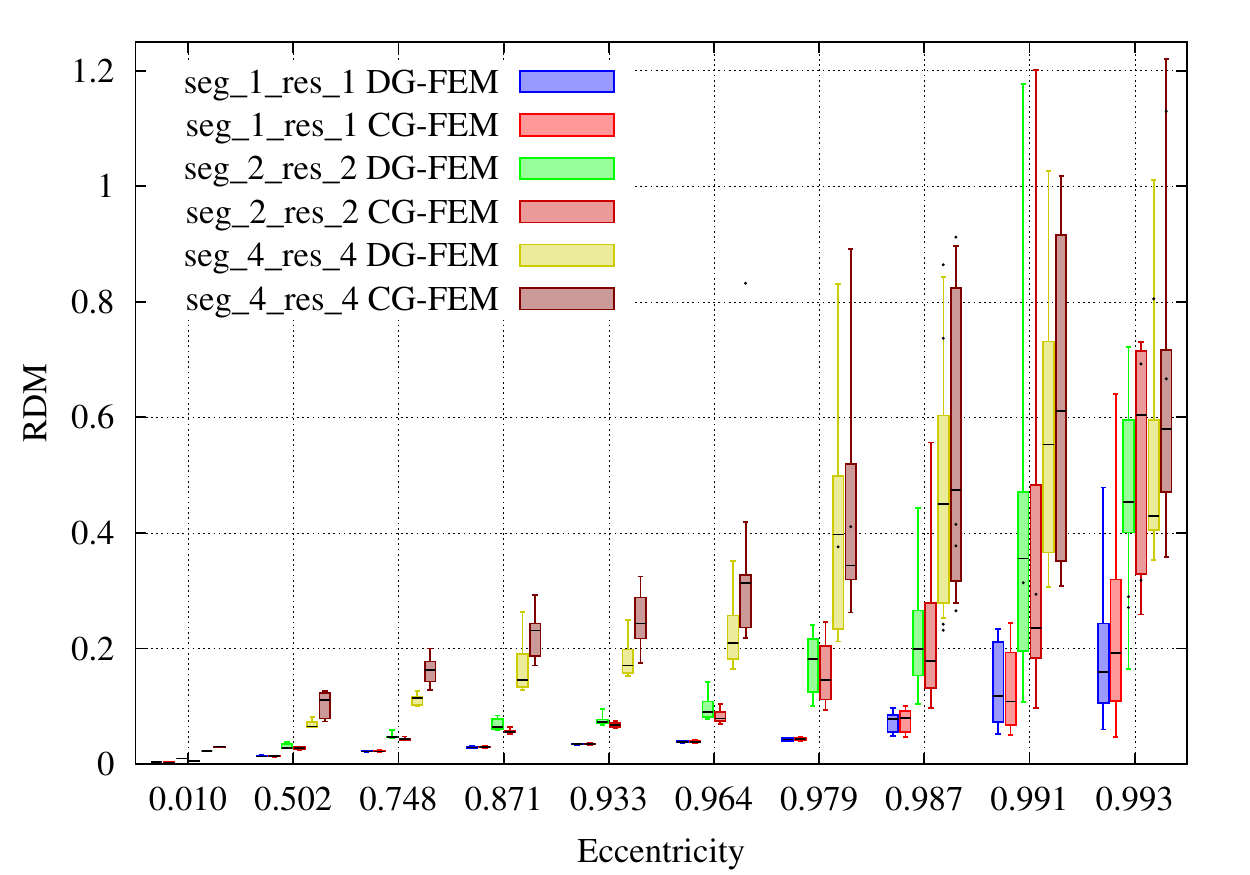}\\%
\includegraphics[width=.48\textwidth]{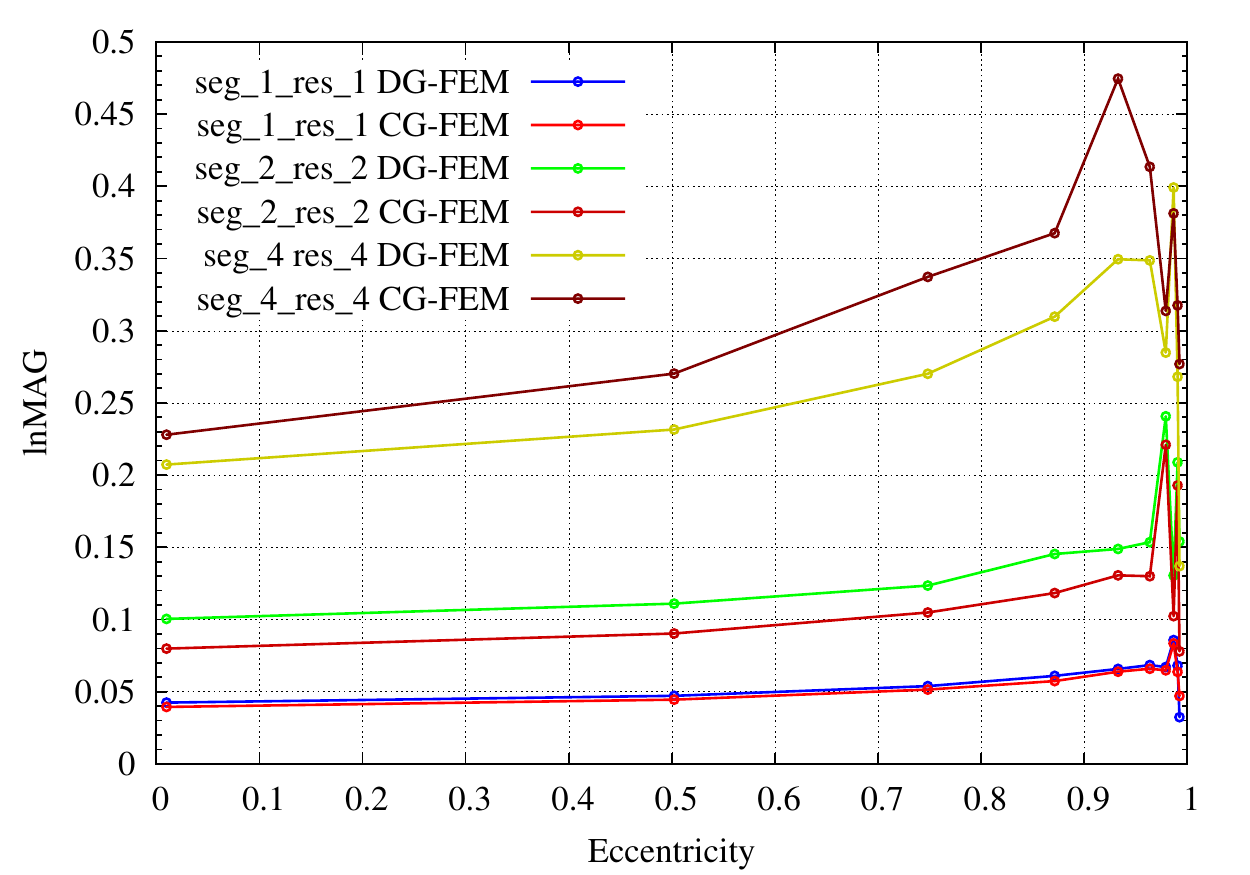} \hfill \includegraphics[width=.48\textwidth]{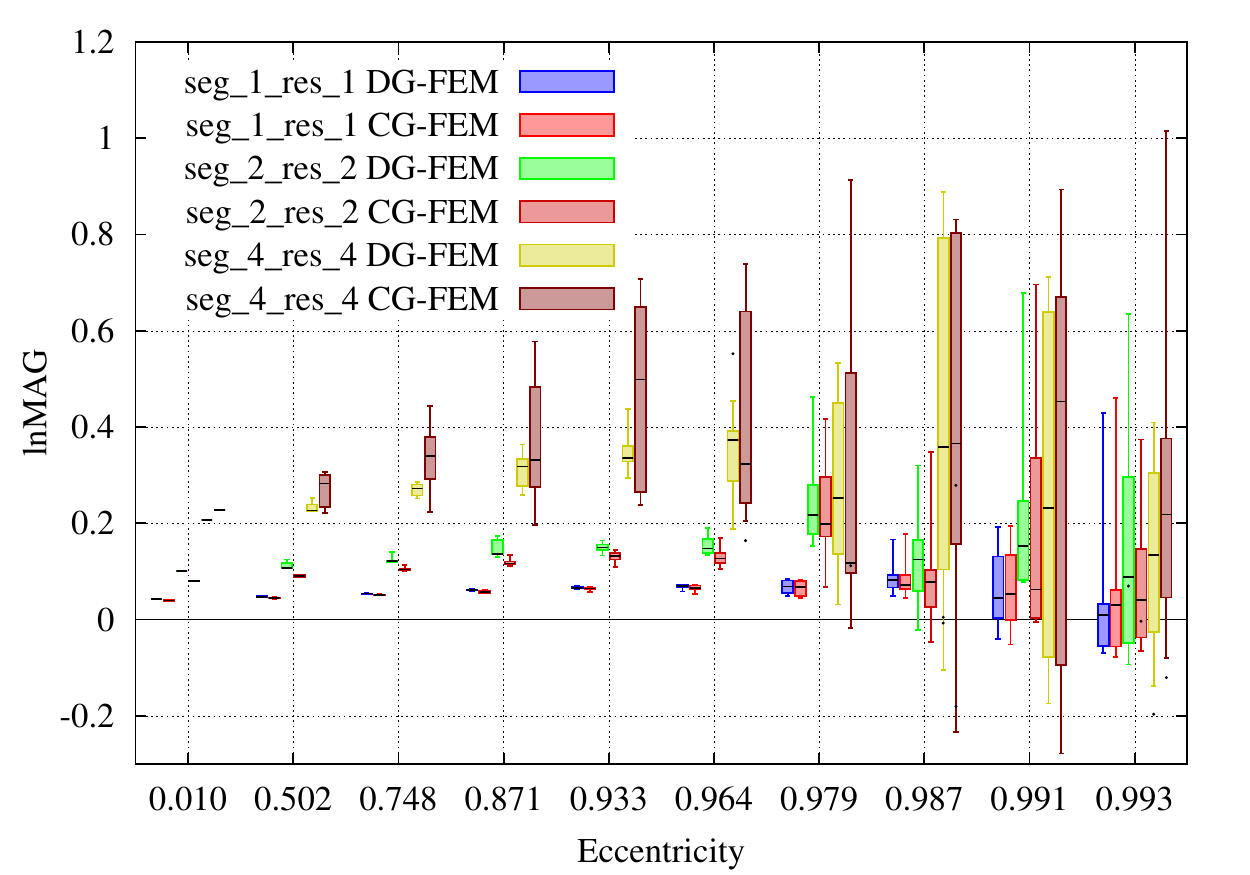}
\caption{Convergence for both the CG-FEM and DG-FEM with increasing mesh and segmentation resolution. Results of radial dipole computations. Visualized are the mean error (left column) and boxplots (right column) of the RDM (top row) and lnMAG (bottom row). Dipole positions that are outside the brain compartment in the discretized models are marked as dots.
Note the different scaling of the x-axes.}%
\label{fig:radial-conv}%
\end{figure*}

\begin{figure*}[!t]%
\centering
\includegraphics[width=.48\textwidth]{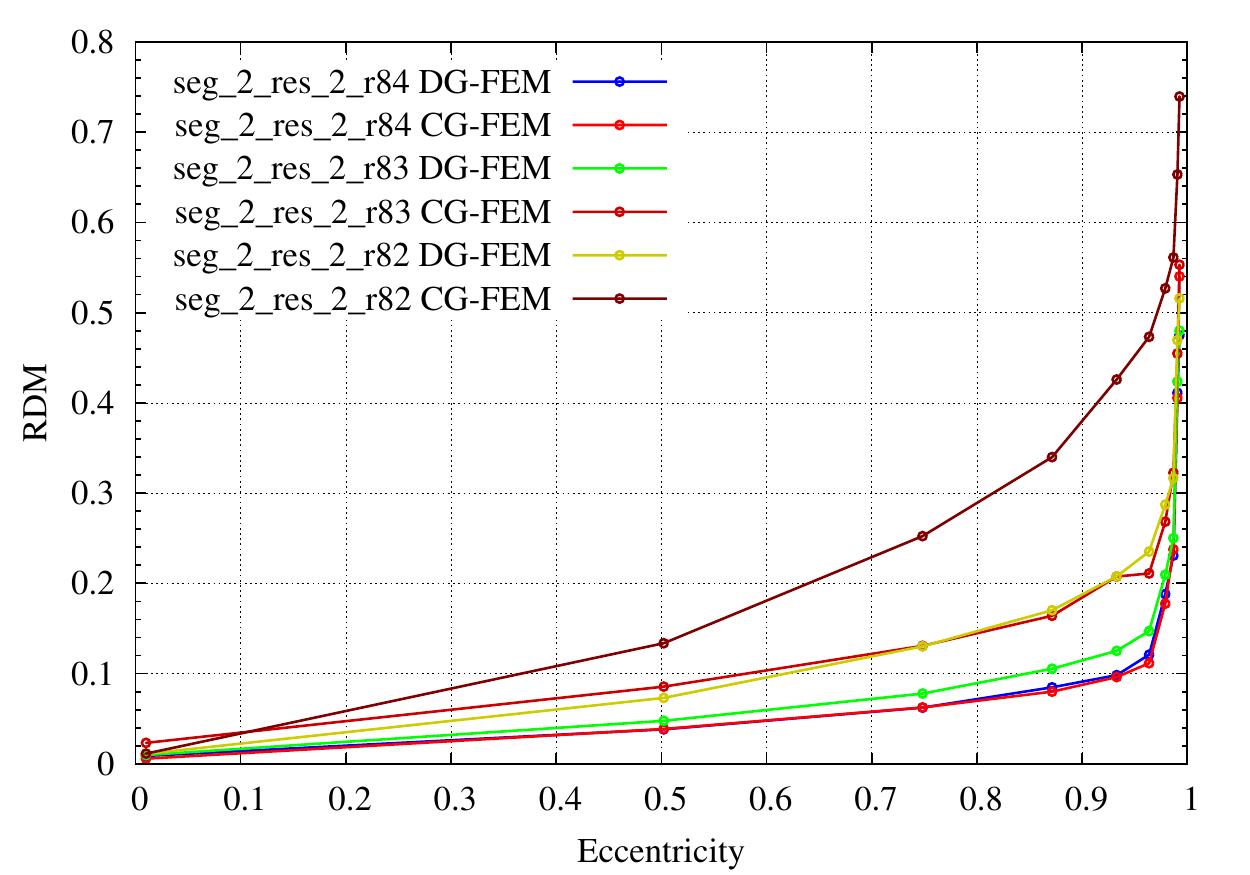} \hfill \includegraphics[width=.48\textwidth]{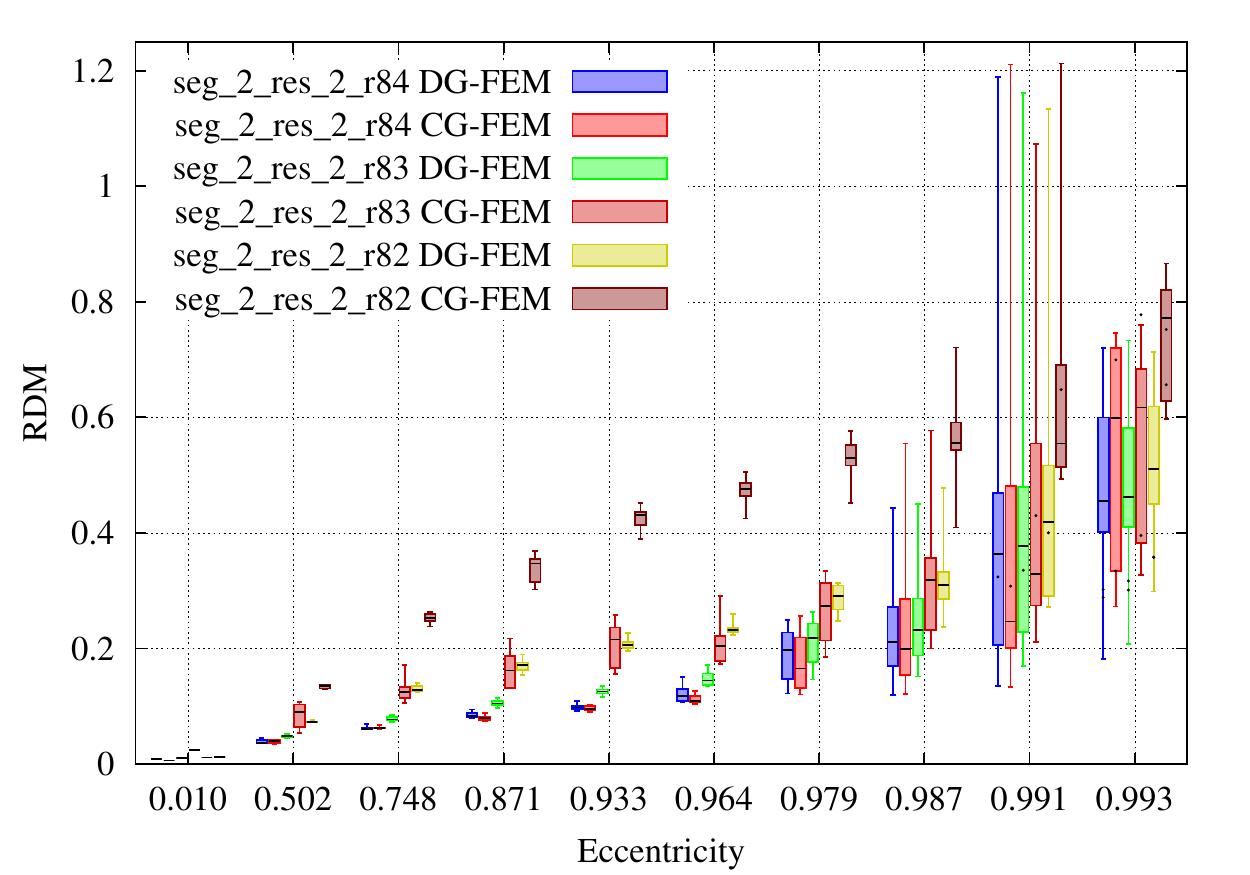}\\%
\includegraphics[width=.48\textwidth]{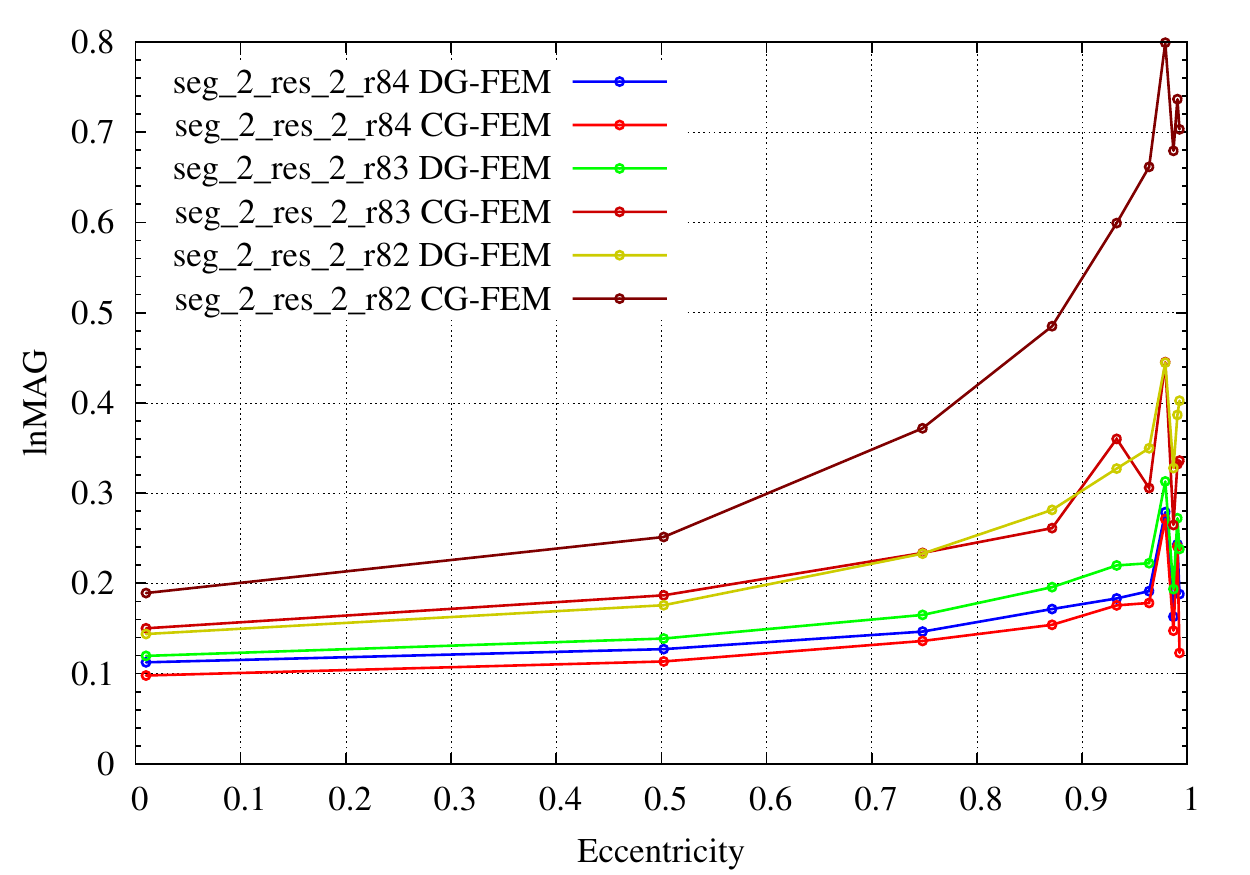} \hfill \includegraphics[width=.48\textwidth]{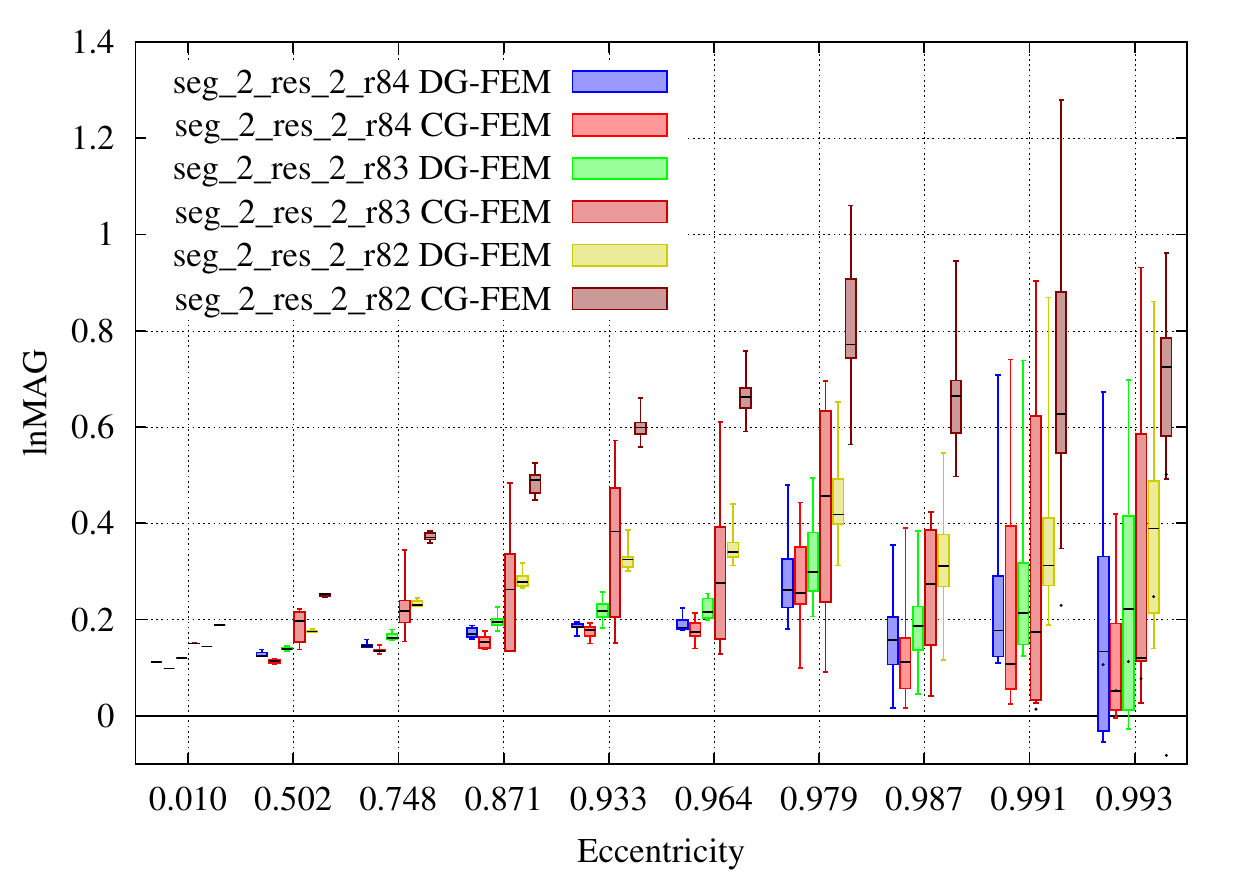}
\caption{Comparison of increase of errors for decreasing skull thickness between the CG-FEM and DG-FEM. Results of radial dipole computations. Visualized are the mean error (left column) and boxplots (right column) of the RDM (top row) and lnMAG (bottom row). Dipole positions that are outside the brain compartment in the discretized models are marked as dots. Note the different scaling of the x-axes.}%
\label{fig:radial-leak}%
\end{figure*}

Figure \ref{fig:radial-geom} shows the convergence of the RDM and lnMAG
errors for the DG method when increasing the segmentation resolution, i.e., improving the representation of the geometry.  Comparing the results for meshes \segres{1}{1},
\segres{2}{2}, and \segres{4}{4} shows the clear
reduction of both
the RDM and lnMAG when increasing mesh and segmentation resolution at the same time. The most accurate model
\segres{1}{1} achieves errors below 0.05 with regard to the RDM for
eccentricities up to 0.979, i.e., a distance of 1.6 mm to the brain/CSF boundary.
For an eccentricity of 0.987, i.e., a
distance of about 1 mm to the brain/CSF boundary, this error increases up
to maximally 0.1.
For even higher eccentricities, the errors clearly
increase up to maximal values of 0.5. However, the median error
clearly stays below 0.2 here and minimal errors are still at about
0.05. The behavior with regard to the lnMAG is very similar, being nearly
constant up to an eccentricity of 0.979, slightly increasing for an
eccentricity of 0.987, and strongly increasing with a high
error variability
for higher eccentricities. The errors for models \segres{2}{2}
and \segres{4}{4} are clearly higher than for model
\segres{1}{1}. However, additionally displaying the
results for the models with refined mesh resolution \segres{2}{1},
\segres{4}{1}, and \segres{4}{2},
where the geometry error, i.e., the error due to the inaccurate representation of the geometry through the segmentation, is kept constant, allows us to estimate whether
the increased errors are due to insufficient numerical accuracy or the
coarse segmentation. We find that both for a
segmentation resolution of 2 mm and 4 mm, the errors are dominated
by the geometry error. Comparing the models with a segmentation
resolution of 2 mm, we find nearly identical errors with regard to the
RDM up to an eccentricity of 0.964
(see right subfigure in Figure \ref{fig:radial-geom}).
Here, the median of the errors remains below 0.1. For higher eccentricities, where sources are
already placed in the outermost layer of elements that still belong to the
brain compartment, the errors for the lower resolved mesh clearly increase faster; the differences are especially large for the two highest eccentricities. With regard to the lnMAG, the effects of the higher mesh resolution are clearly weaker. Even for the outermost sources, notable differences can be seen only due to some outliers,
whereas the medians of the errors stay in a similar range for both mesh
resolutions. For the meshes with a segmentation resolution of 4 mm,
only negligible differences can be seen at all eccentricities;
the medians of the errors are very similar. Differences
can be found only in the maximal values but do not show a systematic
behavior. However, the errors are clearly increased compared to the
models with a higher segmentation resolution, i.e., a better approximation of the geometry. Already at an
eccentricity of about 0.5 the median of the RDM is at about 0.1, increasing to values
above 0.4 for the highest four eccentricities. The same behavior is
observed for the lnMAG, again finding significantly increased errors
compared to the models with a higher segmentation resolution.

In Figure \ref{fig:radial-conv}, the results for the newly proposed DG-FEM are presented side by side
to the CG-FEM for the models \segres{1}{1}, \segres{2}{2}, and
\segres{4}{4}. For the model \segres{1}{1}, the
only notable difference with regard to the RDM can be observed for the
highest eccentricity, where the DG-FEM achieves slightly higher accuracies;
the evaluation of the lnMAG shows even less differences. Also for model
\segres{2}{2}, the two approaches achieve a very similar numerical accuracy for the lower
eccentricities, with RDM errors clearly below 0.1; for eccentricities
between 0.964 and 0.991, the CG-FEM performs slightly better,
whereas for the highest eccentricity the DG-FEM achieves a higher
accuracy, again. However, as analyzed before, the main error source is
the inaccurate representation of the geometry through the segmentation. The lnMAG shows no systematic
difference in accuracy between the two methods in this model. In
the coarsest model, \segres{4}{4}, the DG-FEM performs clearly better
than the CG-FEM even for low eccentricities. Regardless of the
high geometry errors, as seen in Figure \ref{fig:radial-geom},
larger differences in numerical accuracy between the DG- and CG-FEM can be observed for both the RDM and lnMAG
up to an eccentricity of 0.964. For higher eccentricities, possible differences can be less
clearly distinguished due to the dominance of the geometry error and
the resulting generally increased error level.

The most significant accuracy differences between the DG- and CG-FEM
can be seen in Figure \ref{fig:radial-leak}, where we study the increase of errors
for decreasing skull thickness and the resulting increase in the number of skull leakages
(see Table \ref{tab:leaks}). We still find a very similar numerical accuracy for the DG- and CG-FEM in the
leakage-free model \segresR{2}{2}{84} (4 mm skull thickness), as one would expect
given the previous results, but the DG-FEM performs clearly better in the
leaky models \segresR{2}{2}{82}  (2 mm skull thickness) and \segresR{2}{2}{83} (3 mm skull thickness).
Even for low eccentricities, the sensitivity of the CG-FEM to leakages is distinct.
The DG-FEM achieves an only slightly decreased accuracy in the model
\segresR{2}{2}{83} compared to \segresR{2}{2}{84}, which is a first sign that
this approach is clearly less sensitive to leakages. In contrast, the errors of the
CG-FEM for model \segresR{2}{2}{83} are much higher than for model \segresR{2}{2}{84}
(compare \segresR{2}{2}{83} \textit{CG-FEM} with \segresR{2}{2}{84} \textit{CG-FEM}) and already in the range of those of the DG-FEM in the
very leaky model \segresR{2}{2}{82} (compare \segresR{2}{2}{83} \textit{CG-FEM} with \segresR{2}{2}{82} \textit{DG-FEM}). Overall,
we find that the DG-FEM achieves a significantly higher
numerical accuracy than the CG-FEM already for low eccentricities in the leaky models, both with regard to the RDM and lnMAG.

\begin{figure*}[!t]%
\centering
\setlength{\tabcolsep}{.005\textwidth}
\begin{tabular}{ccccl}
& \textit{Geometry} & \textit{CG-FEM} & \textit{DG-FEM} & \\
\raisebox{-0.5\height}{\rotatebox{90}{\textit{\footnotesize \segresR{2}{2}{82}}}} &
\raisebox{-0.5\height}{\includegraphics[width=.25\textwidth]{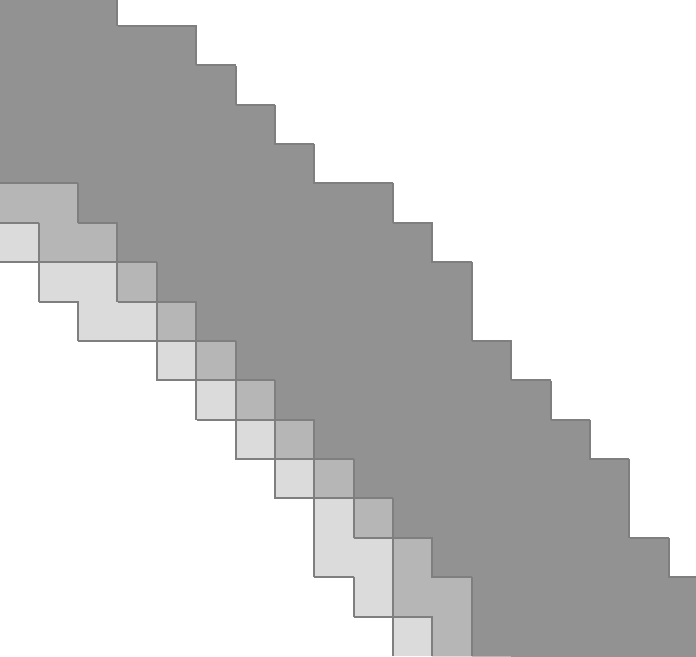}} &
\raisebox{-0.5\height}{\includegraphics[width=.25\textwidth]{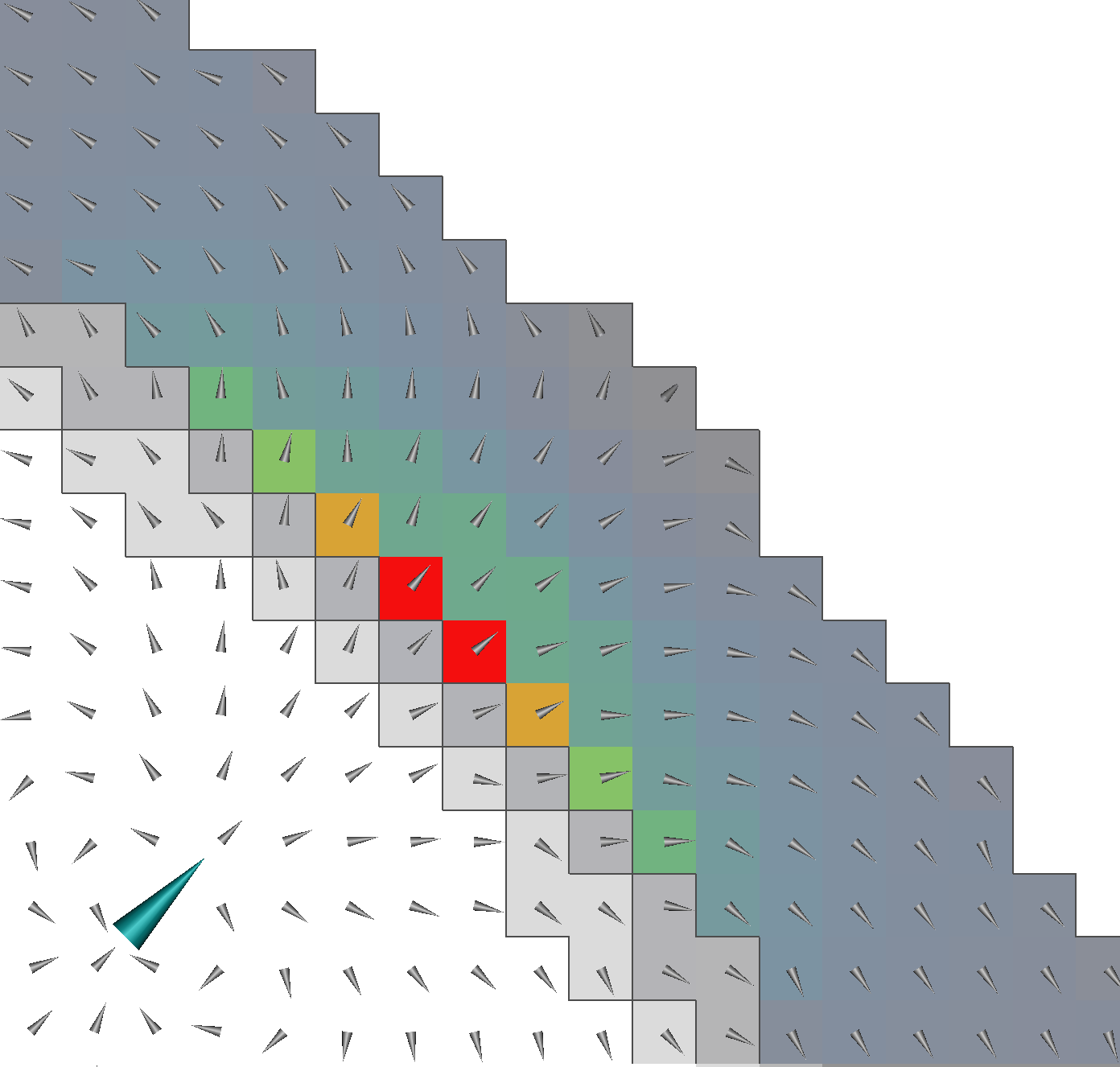}} &
\raisebox{-0.5\height}{\includegraphics[width=.25\textwidth]{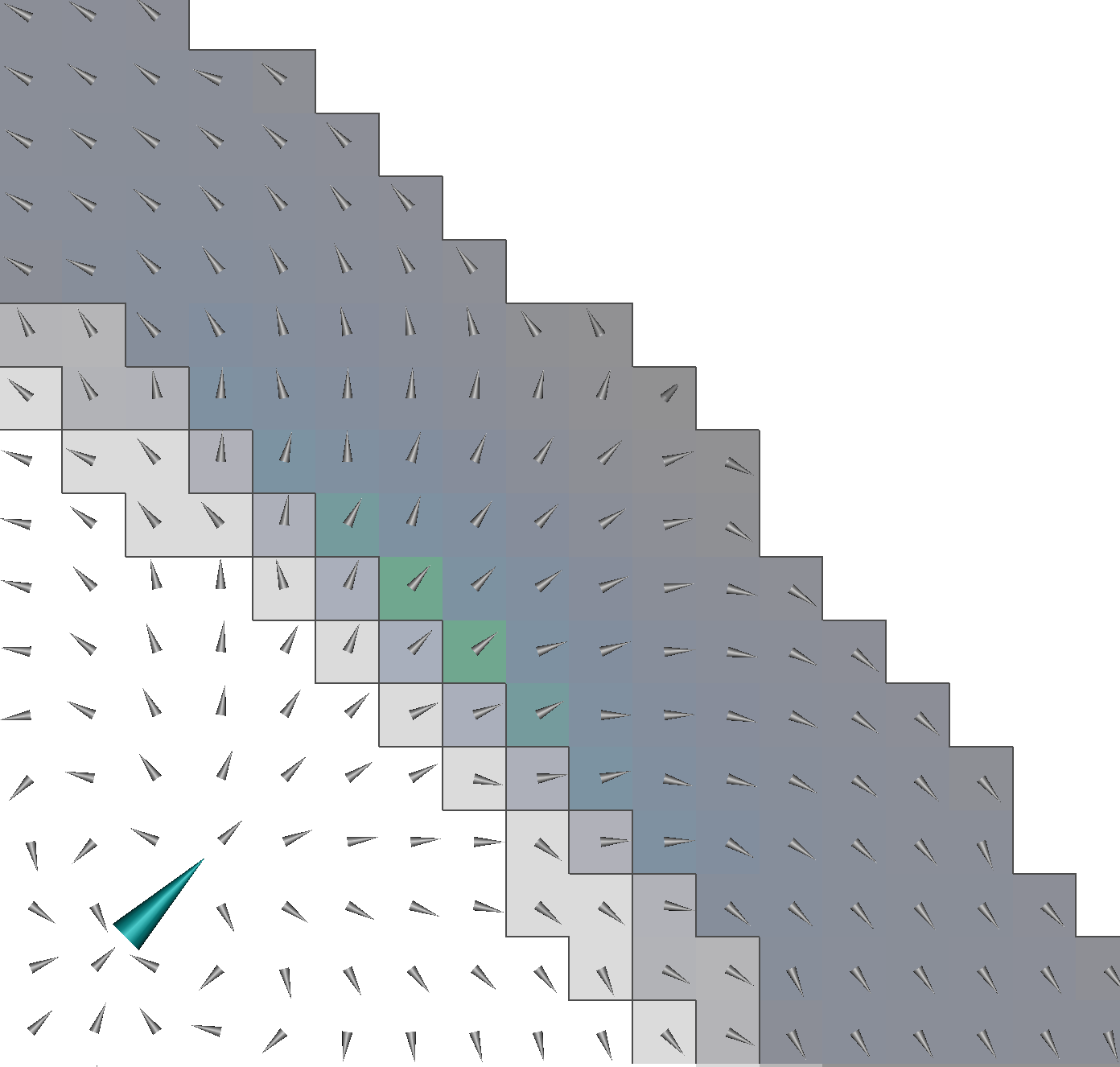}} &
\raisebox{-0.5\height}{\includegraphics[width=.05\textwidth]{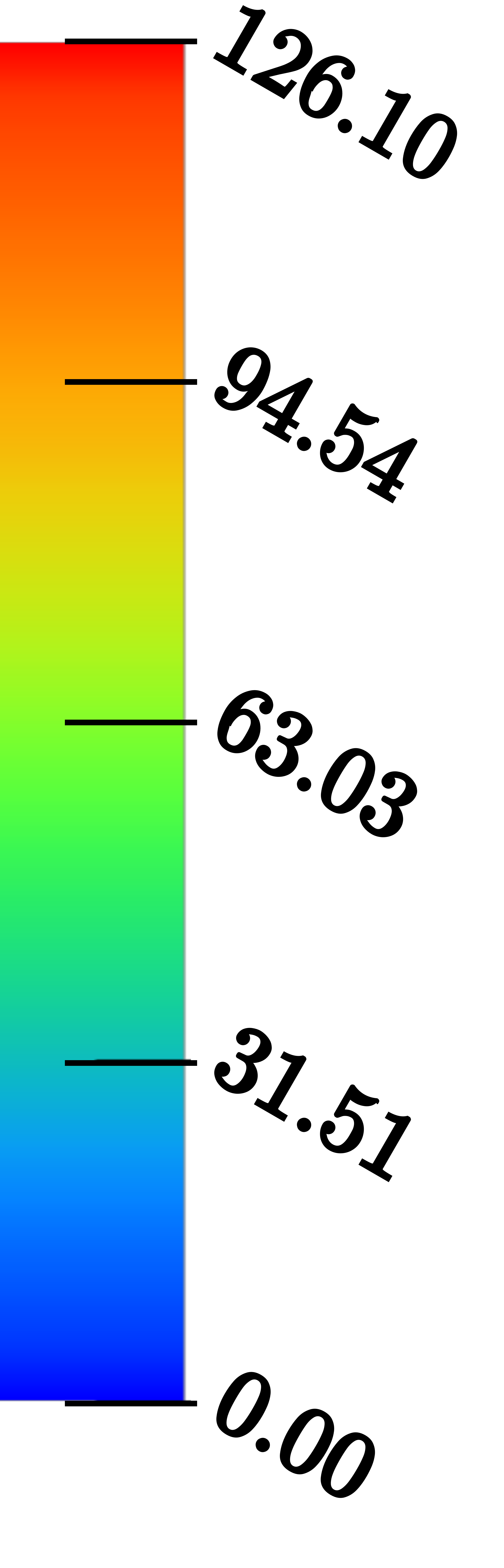}}\\
\raisebox{-0.5\height}{\rotatebox{90}{\textit{\footnotesize \segresR{2}{2}{83}}}} &
\raisebox{-0.5\height}{\includegraphics[width=.25\textwidth]{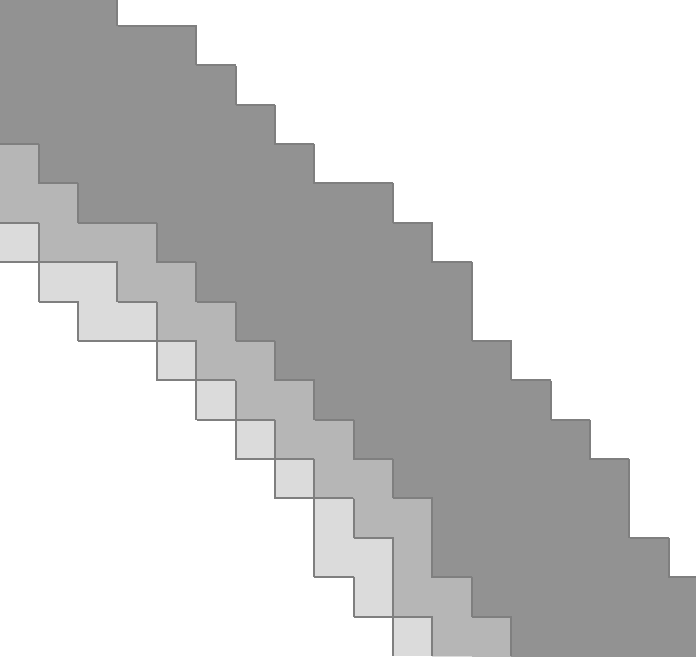}} &
\raisebox{-0.5\height}{\includegraphics[width=.25\textwidth]{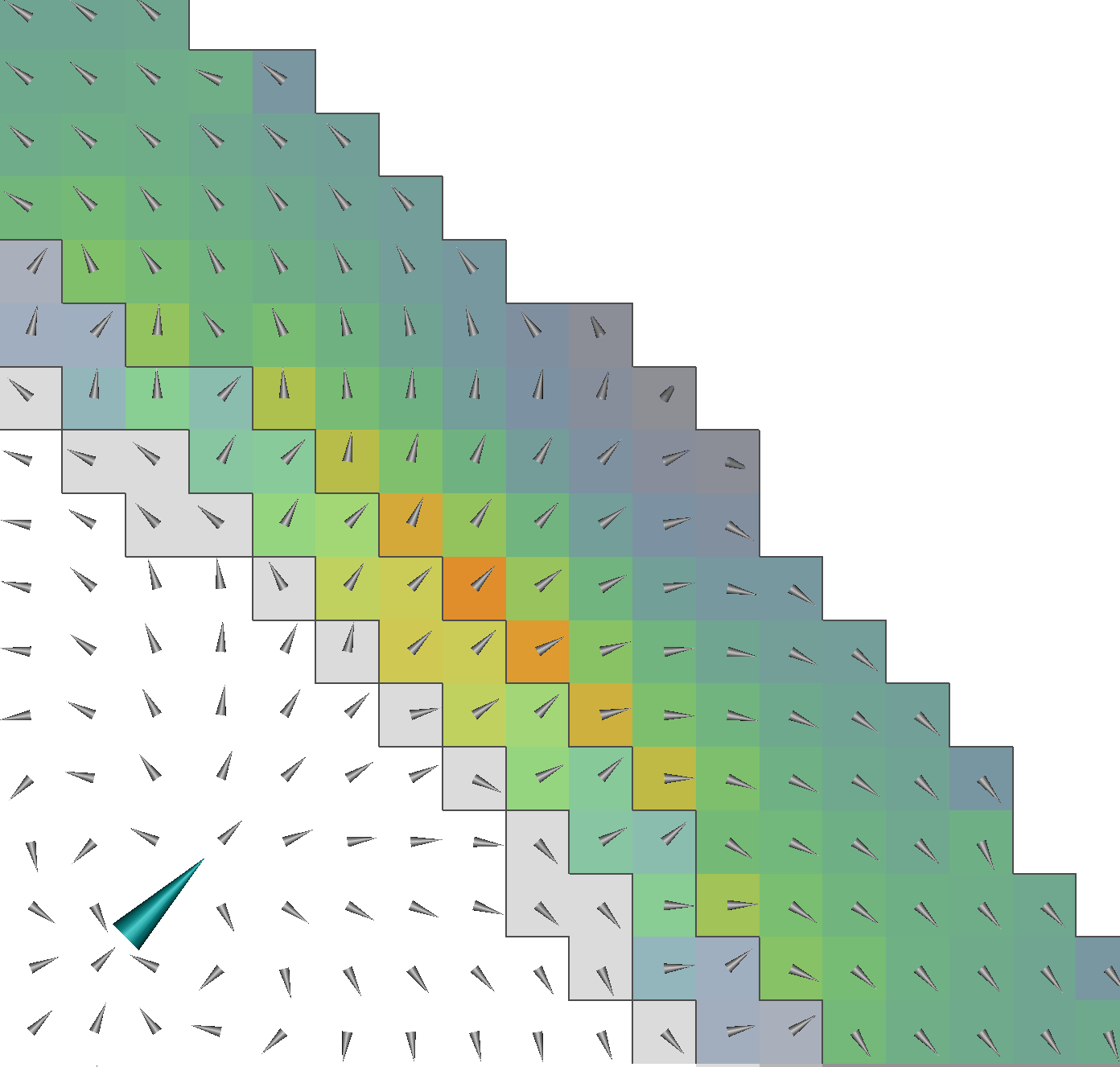}} &
\raisebox{-0.5\height}{\includegraphics[width=.25\textwidth]{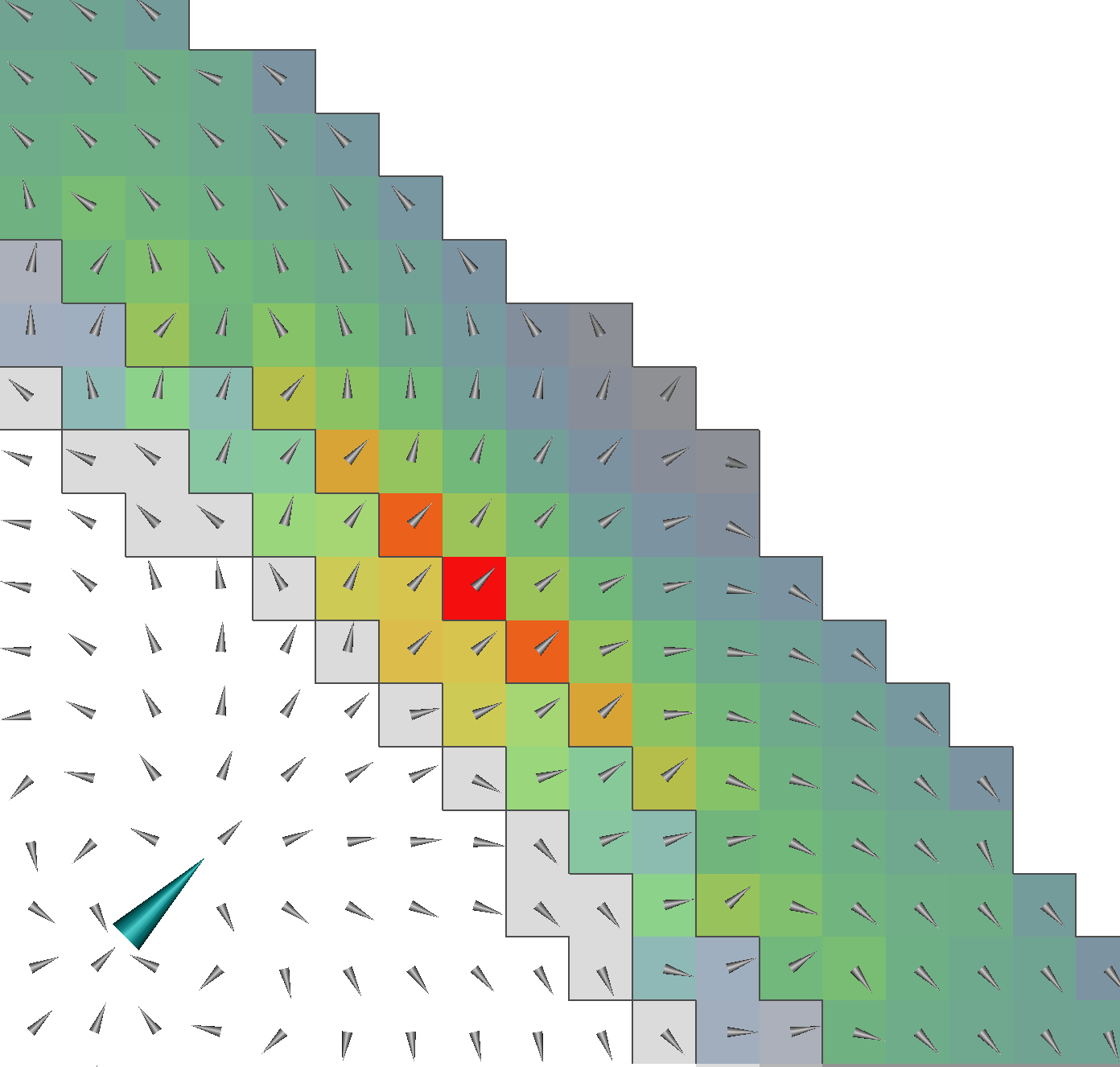}} &
\raisebox{-0.5\height}{\includegraphics[width=.05\textwidth]{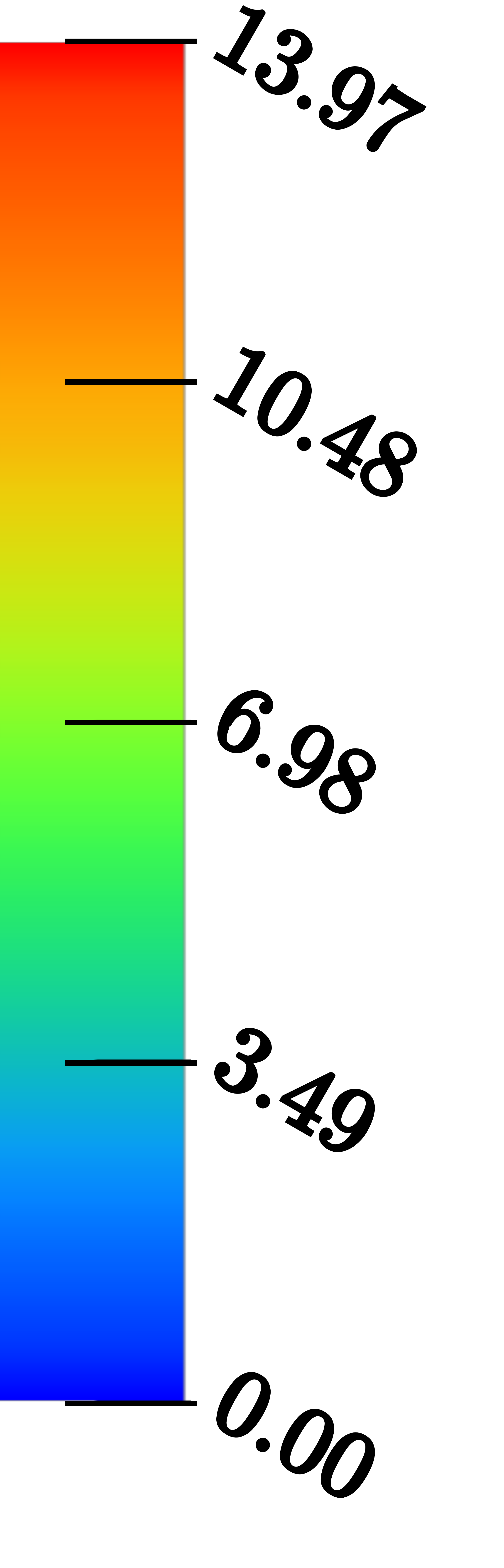}}\\
\raisebox{-0.5\height}{\rotatebox{90}{\textit{\footnotesize \segresR{2}{2}{84}}}} &
\raisebox{-0.5\height}{\includegraphics[width=.25\textwidth]{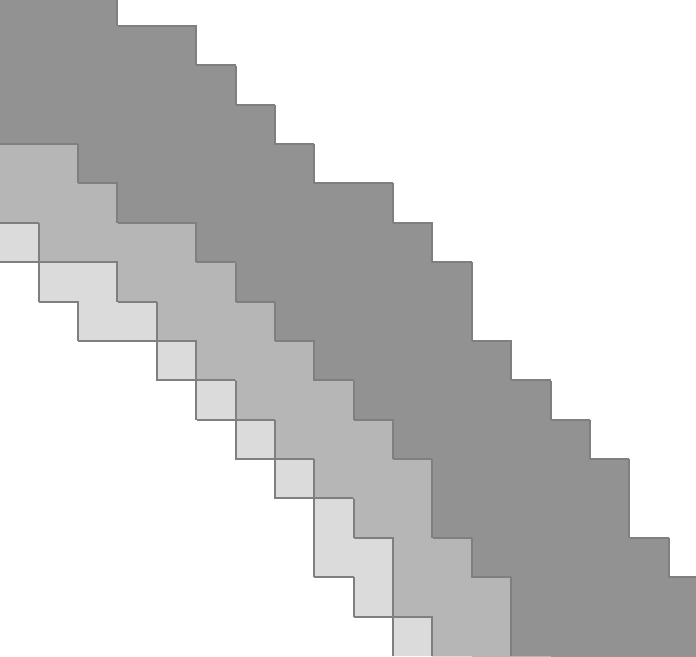}} &
\raisebox{-0.5\height}{\includegraphics[width=.25\textwidth]{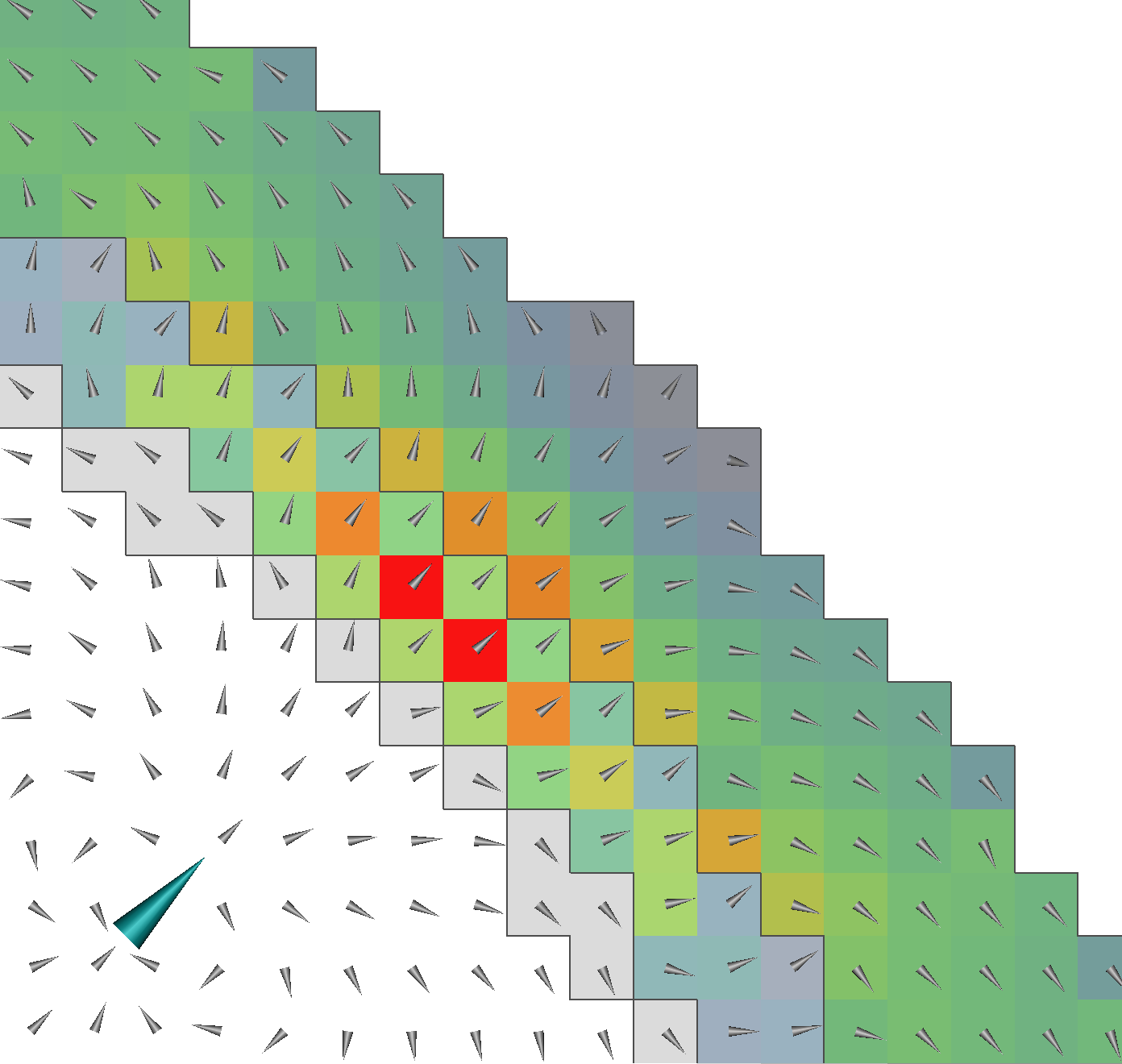}} &
\raisebox{-0.5\height}{\includegraphics[width=.25\textwidth]{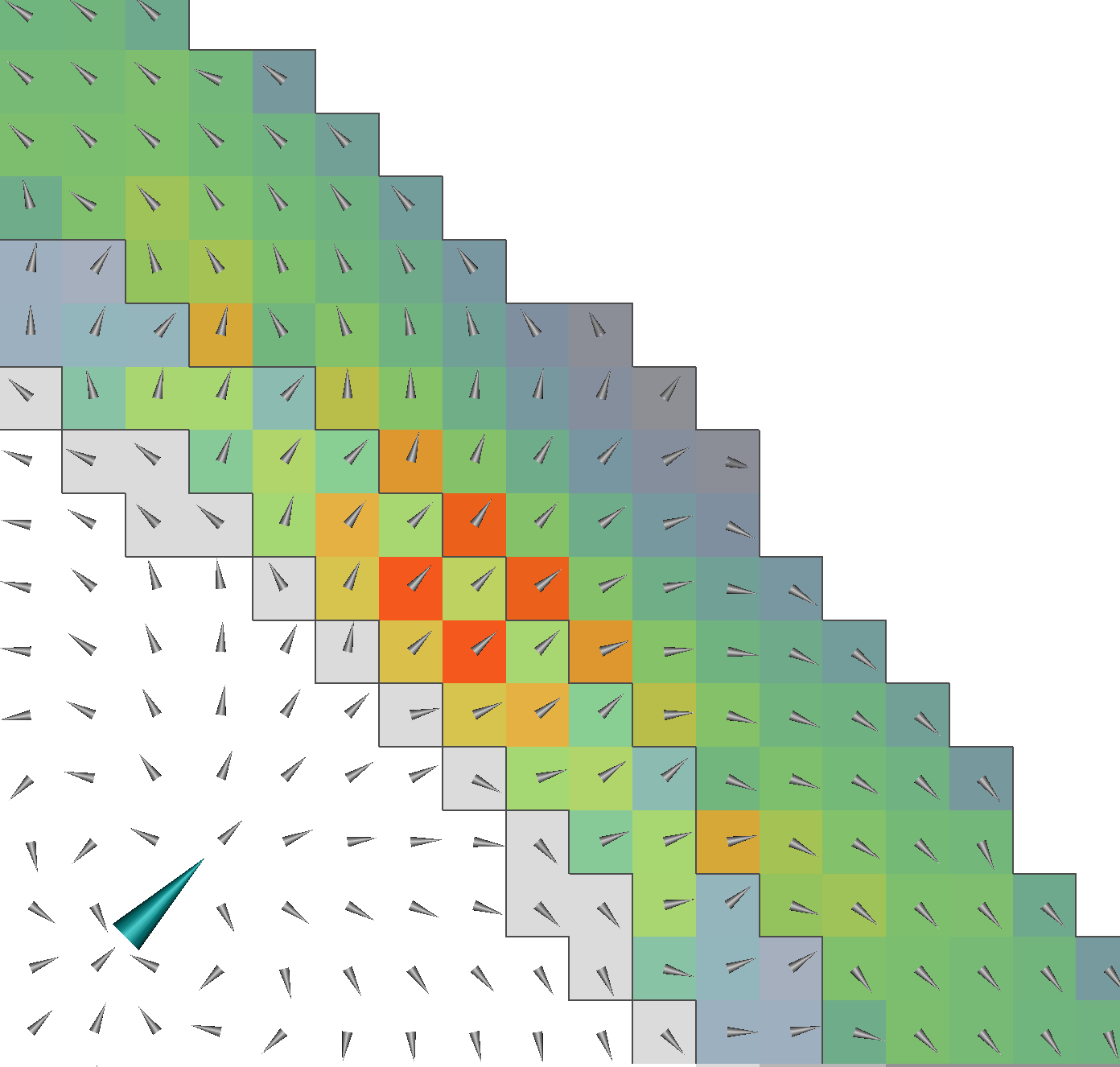}} &
\raisebox{-0.5\height}{\includegraphics[width=.05\textwidth]{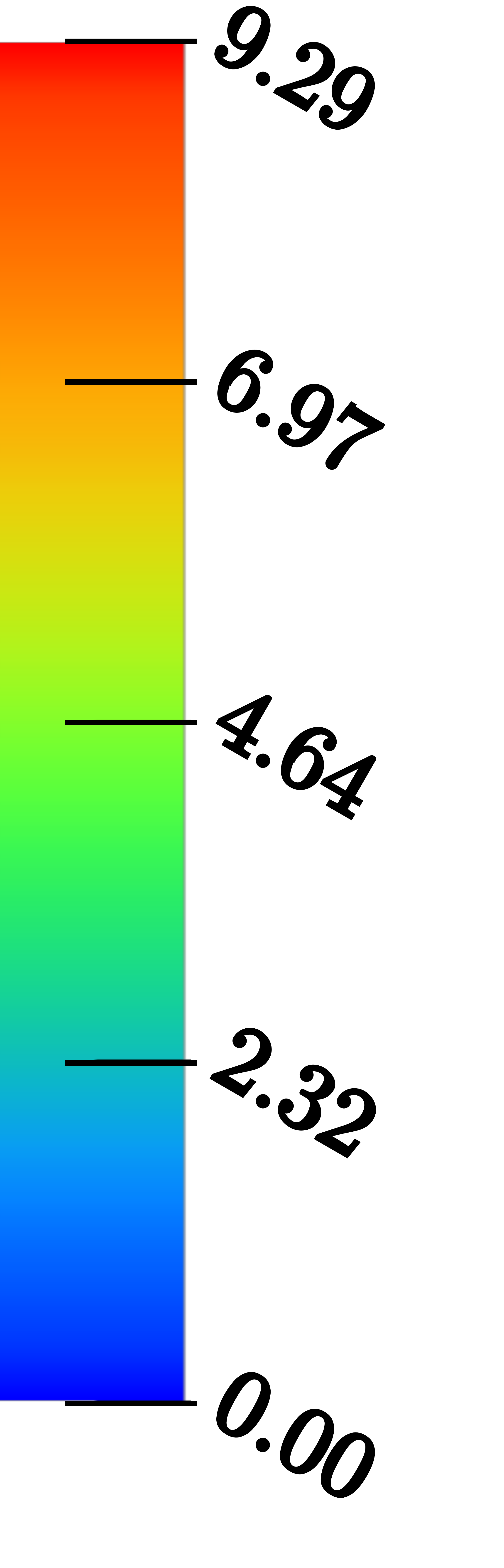}}
\end{tabular}
\caption{Visualization of model geometry (left column), current direction and strength for the CG-FEM (middle column) and DG-FEM (right column) for models \segresR{2}{2}{82} (top row), \segresR{2}{2}{83}  (middle row), and \segresR{2}{2}{84}  (bottom row). The left column shows the model geometry, interior to exterior from bottom left to top right, brain in white, CSF, skull and skin in increasingly dark gray, and air in white. Dark gray lines mark compartment boundaries. In the middle and right columns, the large turquoise cone presents the dipole source. The small and normalized gray cones show the directions of the current flow and, for elements belonging to skull and skin compartments, the coloring indicates the current strength. For each model, the color scale is kept constant for both approaches.}%
\label{fig:vis-current}%
\end{figure*}

\begin{figure*}[!t]%
\centering
\setlength{\tabcolsep}{.005\textwidth}
\begin{tabular}{cccl}
\segresR{2}{2}{82} & \segresR{2}{2}{83} & \segresR{2}{2}{84} & \\
\includegraphics[width=.3\textwidth]{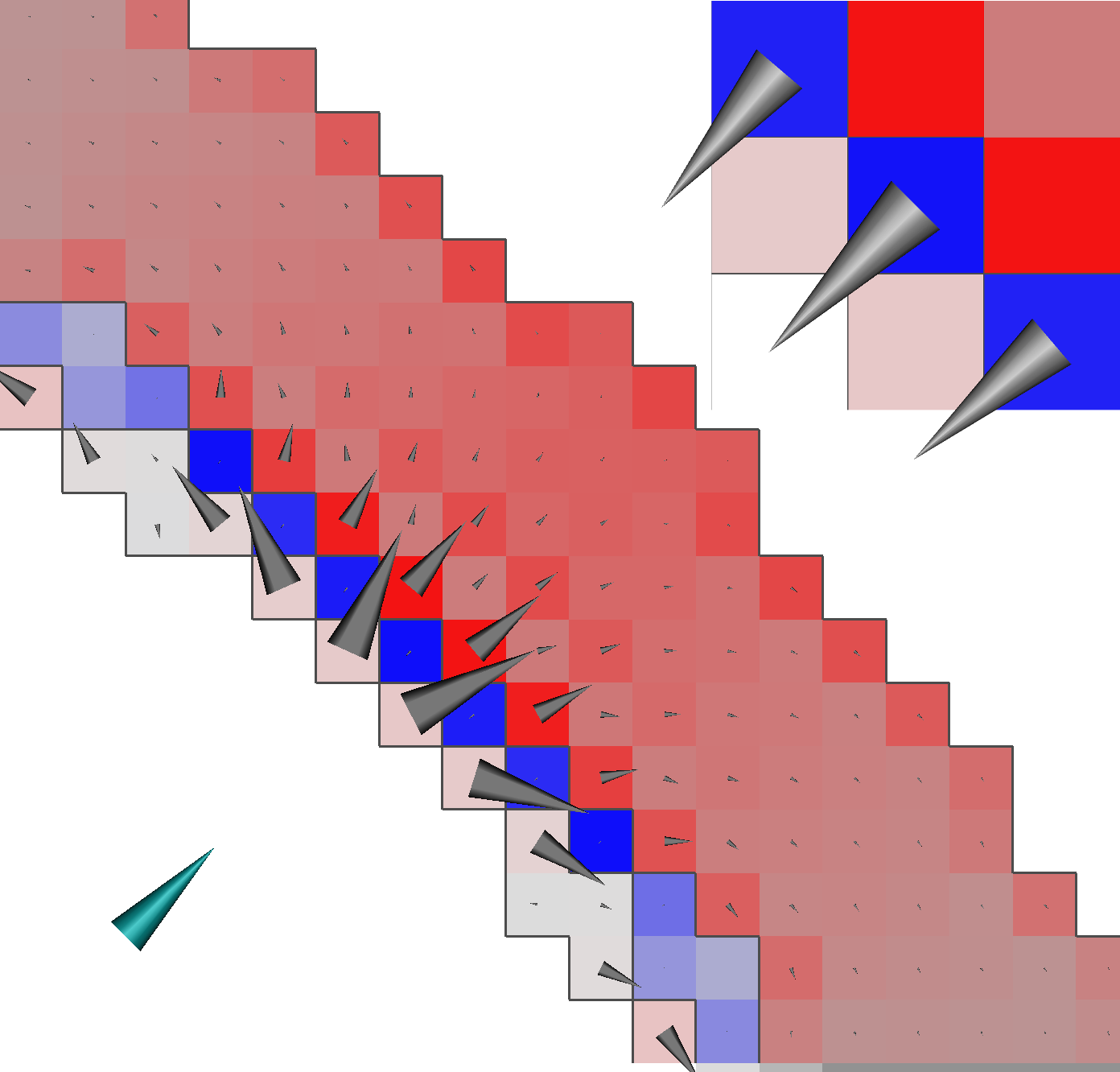} & \includegraphics[width=.3\textwidth]{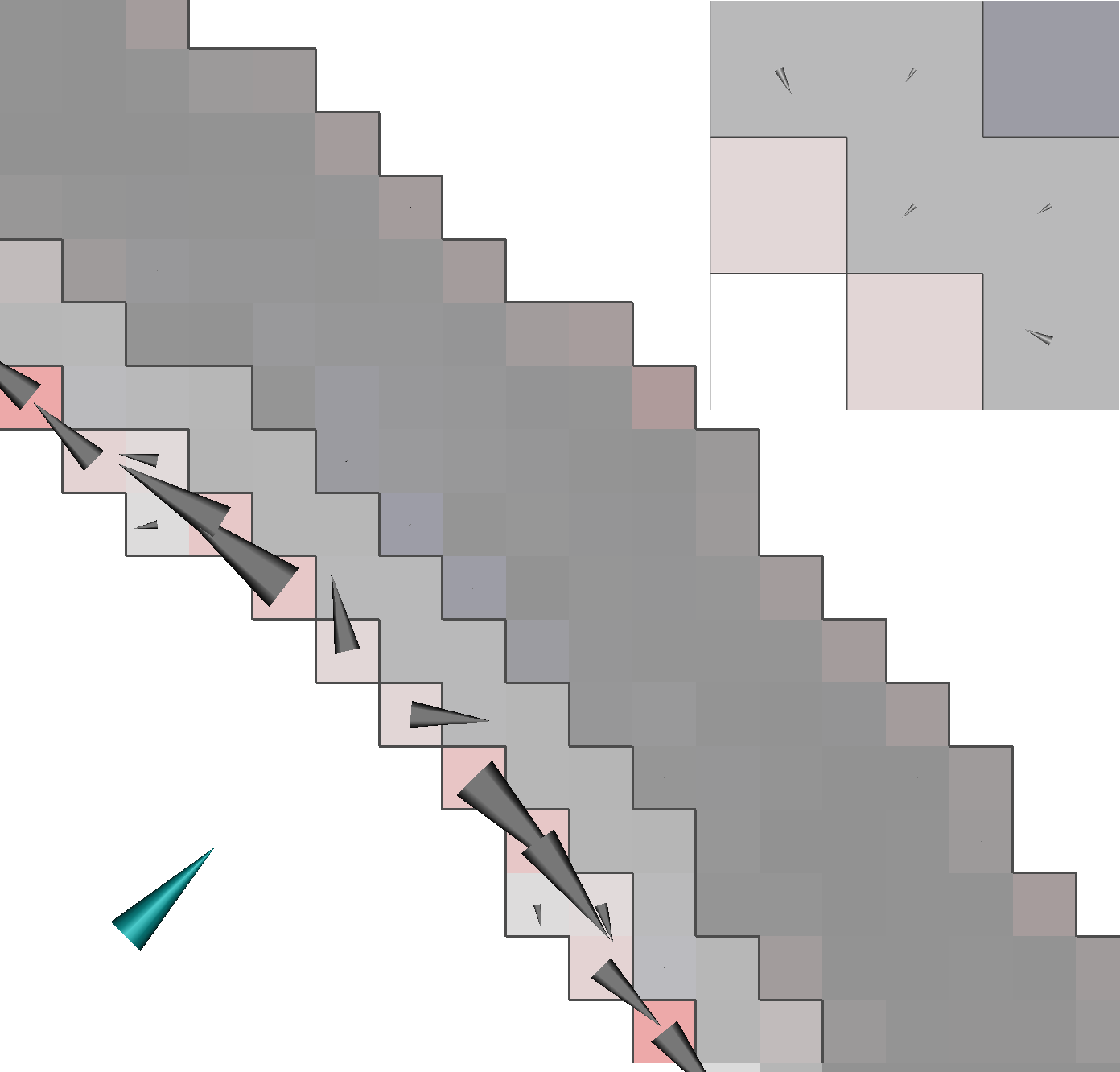} & \includegraphics[width=.3\textwidth]{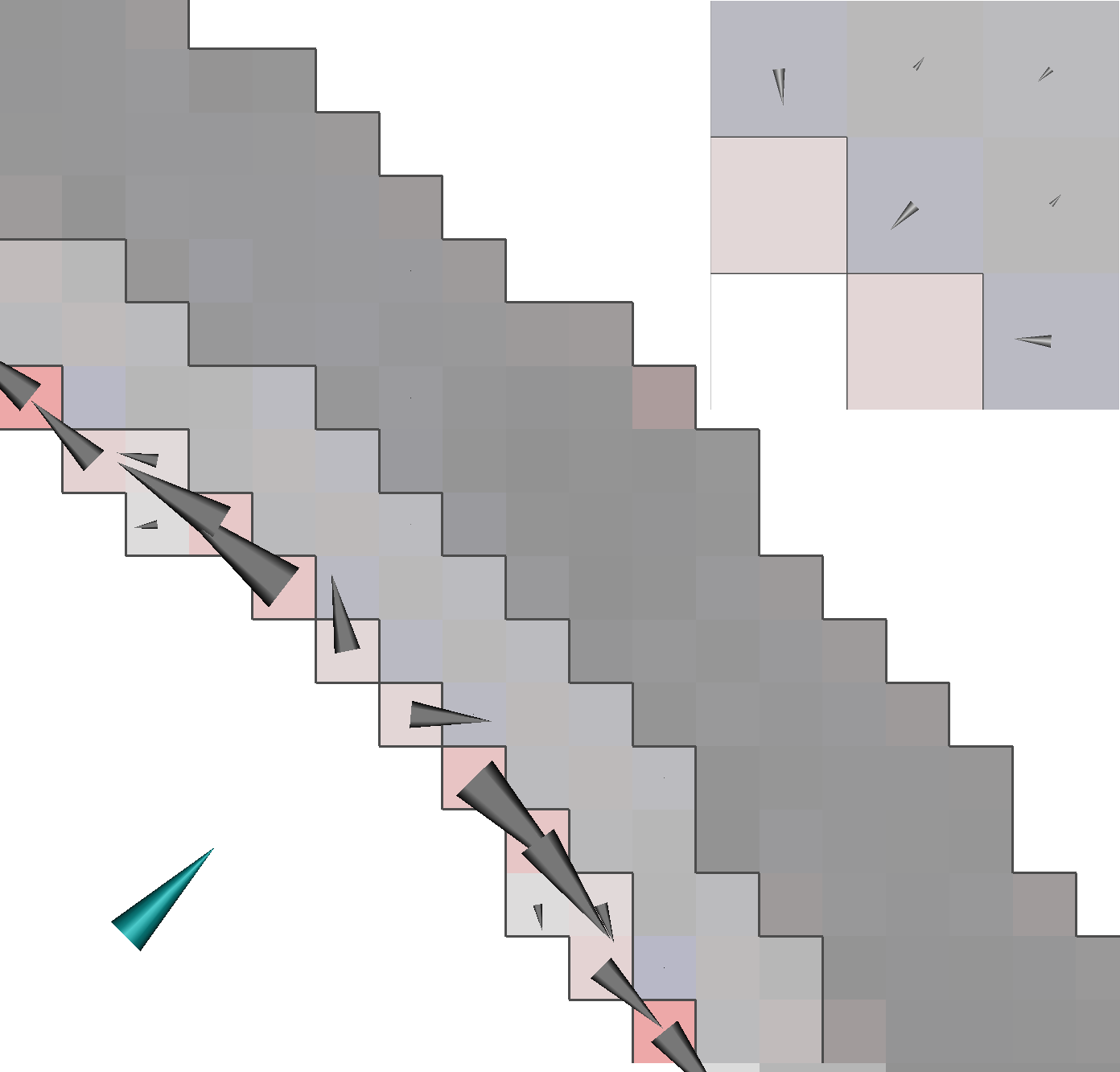} & \raisebox{0.25\height}{\includegraphics[width=.06\textwidth]{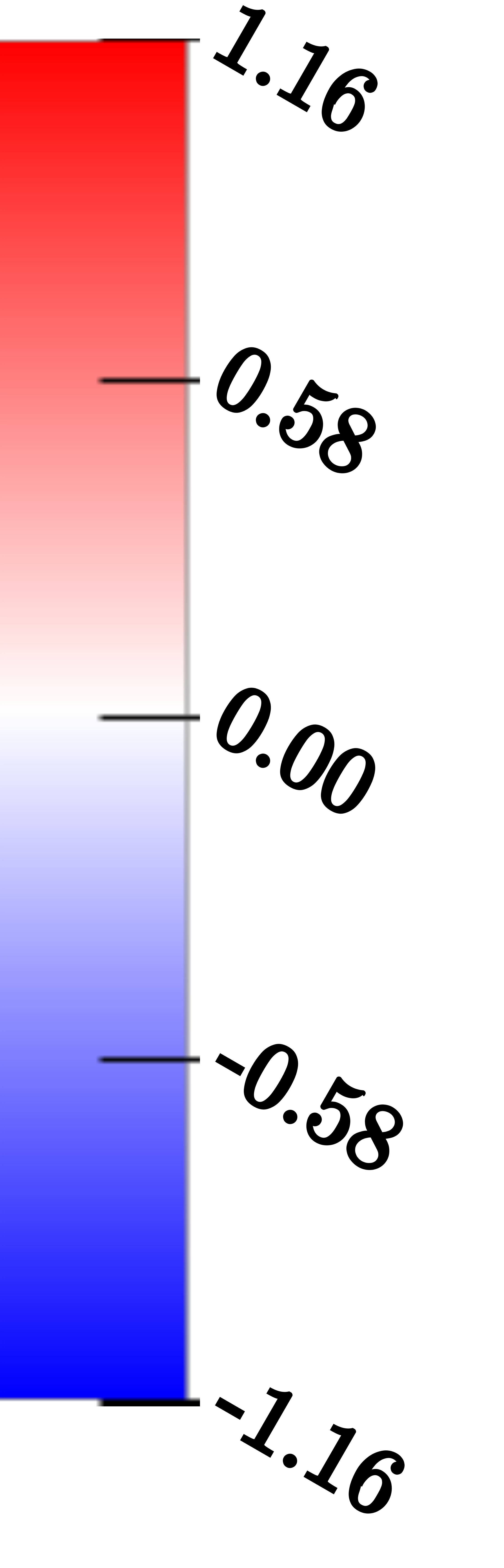}}
\end{tabular}
\caption{Visualization of current flow differences between the CG- and DG-FEM in models \segresR{2}{2}{82} (left), \segresR{2}{2}{83} (middle), and \segresR{2}{2}{84} (right). The turquoise cone presents the dipole source. The coloring shows the \lnMAGj (increase/decrease of the current strength simulated with the CG- compared to the DG-FEM solution).
For all models the maximum of the color scale is chosen as the maximal value in the skin and skull compartment. Gray cones, having the same linear scaling for all models, show the \totDIFFj (difference in current flow). In models \segresR{2}{2}{83} and \segresR{2}{2}{84}, the arrows in skin and skull are not visible due to the relatively small values. Dark gray
  lines mark compartment boundaries. In the top right corner of each subfigure, a detail of the skull elements to which the dipole is pointing is shown. The size of the cones is magnified by a factor of 20 compared to the full image, and only the cones for the skull compartment are visualized.}%
\label{fig:vis-leak}%
\end{figure*}

To illustrate the effect of skull leakages, we generated the
visualizations shown in Figures \ref{fig:vis-current} and
\ref{fig:vis-leak}. In Figure \ref{fig:vis-current}, the electric current direction
and strength for a radial dipole with fixed position and orientation (turquoise cone in the middle and right columns) in the
models \segresR{2}{2}{82} (top row), \segresR{2}{2}{83} (middle row), and
\segresR{2}{2}{84} (bottom row) and with the two numerical approaches CG-FEM (middle column)
and DG-FEM (right column) are visualized. When using the CG-FEM in the model with the thinnest (2 mm) skull compartment,
\segresR{2}{2}{82}, we find extremely strong currents in the innermost
layer of skin elements, i.e., at the interface to the skull.
This effect is especially distinct for the elements to which the dipole is nearly directly pointing. In comparison,
the current strengths found in the skull
compartment are negligible, which is a clear sign for a current leakage
through the vertices shared between the CSF and skin compartment, bypassing the thin
and leaky skull compartment. For the DG approach, these extreme peaks are not
found, and the maximal current strength amounts to only about 30\% of that
of the CG approach. In the other two models (note the much lower scaling in the
middle and lower rows in Figure \ref{fig:vis-current}), we find a clear decrease
of the current strength in the skin compartment compared to the
\segresR{2}{2}{82} model. In these two models and with the given source scenario, none of the approaches seems
to be obviously affected by skull leakage.
However, in model \segresR{2}{2}{83} (middle row), the DG approach shows about 20\% higher peak
currents in the innermost layer of skin elements compared to the CG-FEM. In model \segresR{2}{2}{84} (bottom row), the maximal current strength for the CG-FEM is only about 7\% higher
than for the DG-FEM. The maximal current strength is found in the skull compartment in this model, indicating that no leakage effects occur. If leakage effects would occur, the maximal current strength would be expected to be found in the skin compartment as in the other models. These deviations seem reasonable
considering the relatively coarse resolution of the segmentation, and especially considering the low skull thickness.
The visualizations show that the interplay between source position and direction and
the local mesh geometry strongly influences the local current flow in these
models, leading to current peaks in some elements while neighboring
elements show relatively low currents, as is clearly visible in
model \segresR{2}{2}{84}. In this model, we find strong currents in the two
skull elements connecting the CSF and skin to which the dipole is pointing (see also outward pointing arrows in these elements in the detail in Figure \ref{fig:vis-leak}, right). Locally, these constitute the ``path of least resistance'' between the CSF and skin compartment.

In Figure \ref{fig:vis-leak}, the two measures \lnMAGj and \totDIFFj are visualized to show the differences between the two
methods even more clearly.
As Figure \ref{fig:vis-current} suggests, we find for model \segresR{2}{2}{82} that for the CG-FEM, the current
strength is clearly higher  than for the DG-FEM in those elements of the innermost
layer of the skin compartment that share a vertex with the CSF
compartment, indicated by the high \lnMAGj (red coloring). The visualization of the \totDIFFj (gray arrows) clearly shows that the leakage
generates a strong current from the CSF compartment directly into
the skin compartment that does not exist for the DG-FEM.
At the same time, the \lnMAGj indicates that the current strength in the skull compartment is decreased in the CG-FEM (blue coloring); the detail of the skull elements for model \segresR{2}{2}{82} in Figure \ref{fig:vis-current}, left, actually shows that there is a stronger current
through the skull elements in the DG-FEM than in the CG-FEM simulation (inwards pointing arrows).
We also find high values for the \totDIFFj
in the CSF compartment that are most probably caused
by effects similar to the ``leakage'' effects, i.e., a mixing of
conductivities in boundary elements/vertices.  However, in model \segresR{2}{2}{82}, the color-coding for the \lnMAGj
shows that this is not related to significant relative differences in current strength.
Here, the strongest values for the \lnMAGj are found in the skin
and skull compartment. In turn, for the other two models we find the largest deviations in the CSF compartment, both with regard
to the \totDIFFj and \lnMAGj.  For model \segresR{2}{2}{83},
we furthermore find minor effects with regard to the \lnMAGj, i.e., relative differences of current
strength, in the innermost layer of skin elements, which are also the elements with the highest
absolute current strength among the skin and skull compartment (see also Figure \ref{fig:vis-current}). We also
find slightly increased values for the \lnMAGj in the outermost layer of the skin elements. These
might be artifacts due to the ``staircase''-like geometry of the outer surface
in the regular hexahedral model. However, the \totDIFFj
in the skin and skull is negligible compared to the CSF compartment, and also clearly lower than in model \segresR{2}{2}{82}.
The same holds true for model
\segresR{2}{2}{84}, where the \lnMAGj is
  slightly increased in the skull
and skin compartment, mainly in elements with a small absolute current
strength, as a comparison to Figure \ref{fig:vis-current} shows.
Still, relatively high differences in the \lnMAGj and \totDIFFj are visible in the CSF compartment.
These results indicate that the models \segresR{2}{2}{83} and
\segresR{2}{2}{84} are less affected by skull leakage; the differences are due rather to the different
computational approaches and do not show obvious errors due to the underlying segmentation.

\begin{figure*}[!t]%
\centering
\includegraphics[width=.48\textwidth]{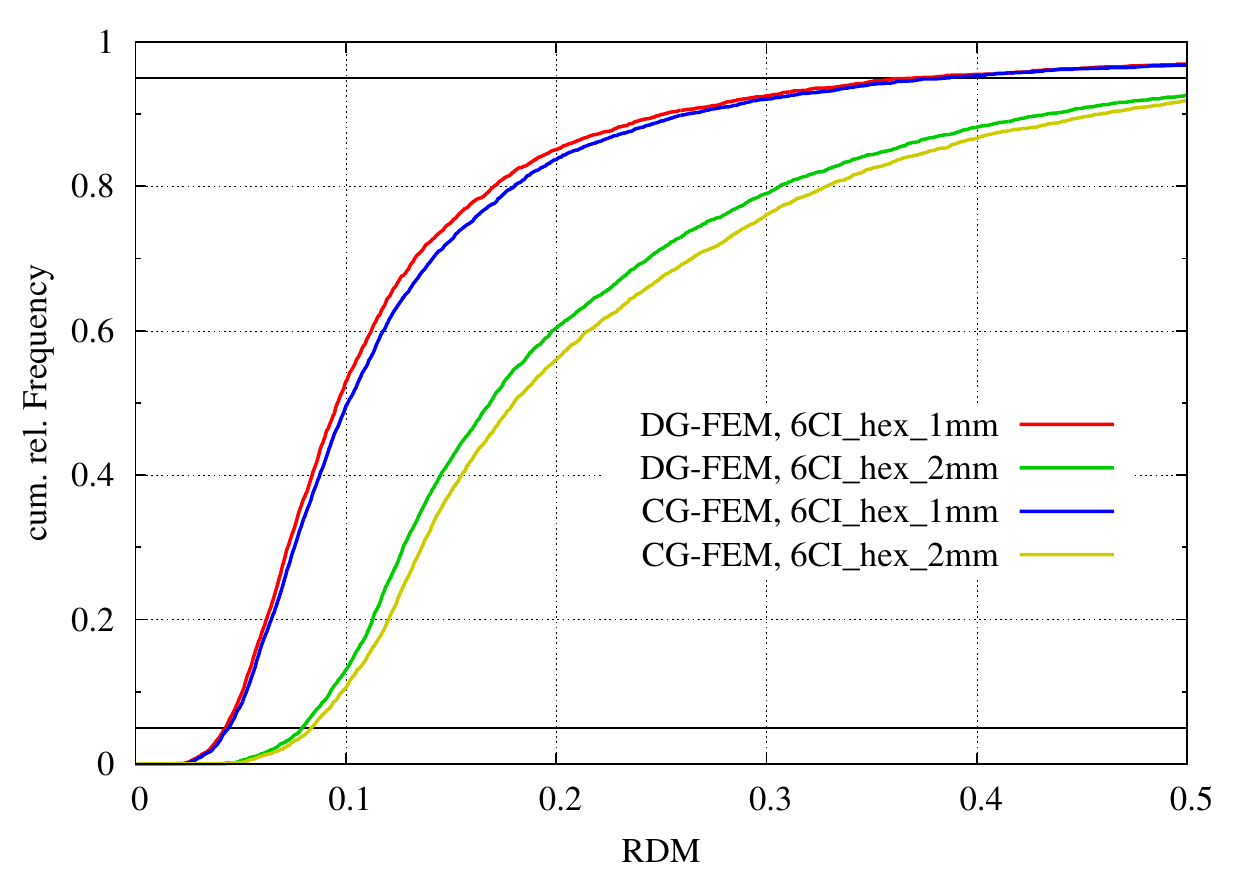} \hfill \includegraphics[width=.48\textwidth]{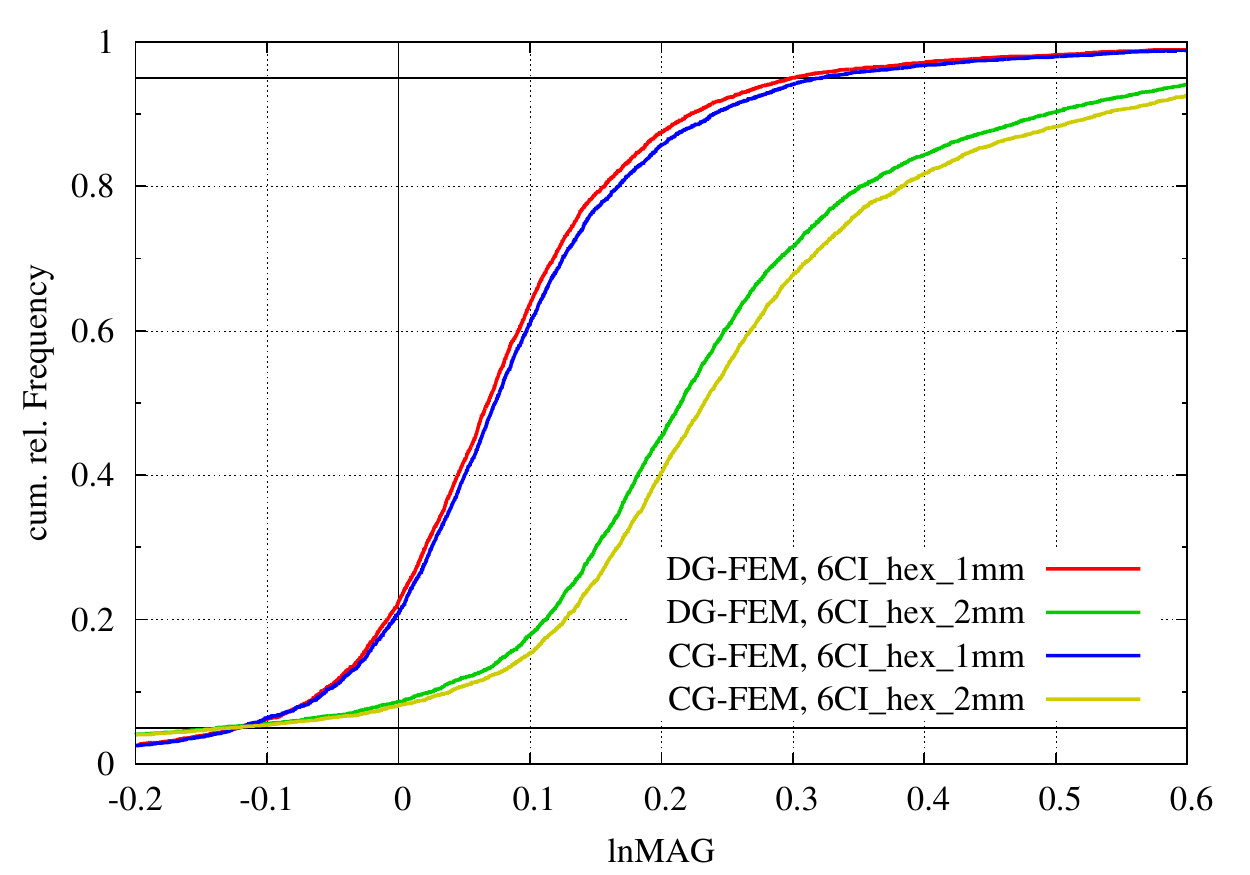}
\caption{Cumulative relative errors of the RDM (left) and lnMAG (right) in realistic six-layer hexahedral head models with 1 and 2 mm mesh resolution in reference to high-resolution tetrahedral model..}%
\label{fig:realistic}%
\end{figure*}

For the realistic head models, \ci{6}{hex}{1mm} and \ci{6}{hex}{2mm}, the RDM and lnMAG in reference to model \ci{6}{tet}{hr} were computed. The cumulative relative frequencies of the RDM and lnMAG are shown in Figure \ref{fig:realistic}. At each position on the x-axis, the corresponding y-value indicates the fraction of sources that have an RDM/lnMAG lower than this value. Accordingly, the rise of the curve should be as steep as possible for both the RDM and lnMAG and furthermore as close as possible to the x=0 line for the lnMAG.

Overall, the results for both the RDM and lnMAG show relatively high errors. This is a consequence of the rather bad approximation of the geometry that is achieved when using regular hexahedra compared to the accuracy that can be achieved using a surface-based tetrahedral model. While the differences between DG- and CG-FEM observed for model \ci{6}{hex}{1mm} are rather subtle, the differences are clear for model \ci{6}{hex}{2mm}. These results mainly underline the observations in the sphere studies.

In model \ci{6}{hex}{1mm}, for both approaches about 50\% of the sources show RDM errors below 0.1 and 95\% of the errors lie below 0.35. The rise of the curve for the DG-FEM is slightly steeper than for the CG-FEM, indicating a higher numerical accuracy, but from an RDM of about 0.3 onwards both curves are nearly overlapping. The RDM errors for the DG- and CG-FEM are clearly increased for the lower mesh resolution. In this model, the difference between the DG- and CG-FEM is also more distinct, e.g., for the DG-FEM more than 60\% of the sources have an RDM below 0.2, but this is only the case for about 56\% of the sources when using the CG-FEM.

The results for the lnMAG are in accordance with those obtained for the RDM. Again, the DG-FEM performs only slightly better than the CG-FEM in model\linebreak \ci{6}{hex}{1mm}, whereas the differences in model \ci{6}{hex}{2mm} are more distinct.

Compared to the results in the sphere models, the differences even in model \ci{6}{hex}{2mm} seem to be rather small. However, it has to be taken into account that the leakages in this model are nearly all located in temporal regions, so that only a fraction of the sources is affected.

\begin{table}[!t]
\renewcommand{\arraystretch}{1.3}
\caption{Computational effort (from left to right): number of unknowns
  (DOFs), solving of a single equation system (${t_\text{solve}}$), overall computation time of transfer matrix (${t_\text{tansfer}}$), setup time of a single right-hand side (${t_\text{rhs}}$), overall computation time of leadfield matrix (${t_\text{lf}}$, 4,724 sources) and the total computation time (${t_\text{total}}$).}
\label{tab:computation}
\centering
\small
\begin{tabular}{lrrrrrr}
\hline
                     & DOFs       & $t_\text{solve}$ & {$t_\text{transfer}$} & {$t_\text{rhs}$} & {$t_\text{lf}$} & {$t_\text{total}$}\\
\hline \hline
\ci{6}{hex}{1mm}, DG & 30,968,232 & 1,468~s    & 115,994~s       & 71~s       & 336,118~s & 452,112~s \\
\ci{6}{hex}{2mm}, DG & 3,876,256  & 136~s      & 10,750~s        & 8.8~s      & 41,939~s  & 52,689~s  \\
\ci{6}{hex}{1mm}, CG & 3,965,968  & 185~s      & 14,634~s        & 68~s       & 321,705~s & 336,339~s \\
\ci{6}{hex}{2mm}, CG & 508,412    & 20~s       & 1,588~s         & 8.7~s      & 41,222~s  & 42,810~s  \\
\hline
\end{tabular}
\end{table}

The computation times for the DG- and CG-FEM in models \ci{6}{hex}{1mm} and \ci{6}{hex}{2mm} are shown in Table \ref{tab:computation}. All computation times are single CPU wall-clock times without exploitation of parallelization or vectoring. The solving time for a single equation system $t_\text{solve}$ grows approximately linear with the number of degrees of freedom. This result corresponds to the theoretically predicted optimal scaling \cite{Braess}. Accordingly, the setup times for the transfer matrices, $t_\text{transfer}$, are clearly higher for the DG-FEM than for the CG-FEM.

In contrast to the solving times, the setup times for a single right-hand side, $t_\text{rhs}$, differ only slightly between the CG- and DG-FEM, being below 10 s for model \ci{6}{hex}{2mm} and around 70 s for model \ci{6}{hex}{1mm}. For the here used source model and the given number of sources, the overall computation time is clearly dominated by the computation of the right-hand sides. This part of the computation takes twice as long as the setup of the transfer matrix even for the DG-FEM and model \ci{6}{hex}{1mm}.

\section{Discussion}
In this paper we presented the theoretical derivation
of the subtraction FE approach for EEG forward
simulations in the framework of discontinuous Galerkin methods. The
scheme is consistent and fulfills a discrete conservation
property. Existence and uniqueness follow
from the coercivity of the bilinear form.

Numerical experiments in sphere models showed the convergence of the DG solution
toward the analytical solution with increasing mesh resolution and
better approximation of the spherical geometry with increasing segmentation resolution. We also showed
that the numerical accuracy of the DG-FEM is dominated by the
geometry error, whereas the actual mesh resolution in a model with a bad geometry approximation due to coarse segmentation resolution had only a minor
influence on the numerical results (Figure \ref{fig:radial-geom}). The
inaccurate representation of the geometry, especially for coarse mesh resolutions, is
visible by the ``staircase-like'' boundaries in Figure \ref{fig:vis-mesh}.

In the comparisons of DG- and the commonly used CG-FEM, we did not find remarkable
differences for models with higher mesh resolutions (1 mm, 2 mm), as
the results in Figure \ref{fig:radial-conv} are in the same range for
both approaches in the models \segres{1}{1} and \segres{2}{2}. In
this set of experiments, three main error sources can be identified: geometry errors, numerical
inaccuracies, and leakage effects.

First, there is the error in the representation of the
geometry as a consequence of approximating the spherical models by voxel segmentations of different resolutions, which is increasing with coarser segmentation resolutions; see also Figure \ref{fig:radial-geom}.
We thus strongly recommend the use of segmentation resolutions, and thereby necessarily MRI resolutions, as high as practically feasible, possibly even locally refined
when zoomed MRI technology is available. In fact, a newly developed zoom technique for MRI has become available for practical use, based on a combination of parallel transmission of excitation pulses and localized excitation \cite{Blasche2012zoomedMRI}.
A first usage of this zoom technique can be found in \mbox{\cite[Chapter 5]{aydin2016,AydinDiss2015}}.
Moreover, in future work, based on \cite{engwer2009:udg}, we plan to further develop a cut-cell
approach that allows for an accurate representation of the geometry
while introducing only a negligible amount of
additional degrees of freedom. Thus, the achieved accuracy can be
increased while the computational effort is hardly affected (see first results in\mbox{\cite{nuessing2016}}).

Second, we have the numerical inaccuracy due to the discretization of
Equation \eqref{eq:forward-strong} in combination with
the strong singularity introduced by the assumption of a point dipole, which
is the main cause for the numerical inaccuracies of the subtraction approach for highest eccentricities, where
the source positions are very close to the next conductivity jump (cf. Figure \ref{fig:radial-conv}).
A rationale for this effect has been given in \cite{CWolt1,FDrec1}.
In future work, we are therefore planning
to adapt other source modeling approaches such as the Venant \cite{CHW:Tou65,CHW:Sch94,HBuch1997,CHW:Wol2007b,JVorw2012},
the partial integration \cite{YYan1991,CHW:Wei2000,CHW:Wol2007b,SVall2010,JVorw2012}, or the Whitney
approach\mbox{\cite{CHW:Tan2005,CHW:Pur2011,pursi2016}} to the DG-FEM framework. Until now, these have been
formulated and evaluated only for the CG-FEM. Compared to the subtraction approach, these approaches have the further advantage of a
strongly decreased computational effort for the setup of the right-hand-side vector \cite{CHW:Wol2007b,JVorw2012}.

The third source of error, the ``leakage effects'', explains the large differences in
numerical accuracy between the CG- and DG-FEM that can be observed in
model \segres{4}{4}. Due to the coarse resolution of the segmentation in comparison to
the thickness of the skull compartment (4 mm segmentation resolution, 6 mm
skull thickness), this model can already be considered as (at least partly)
leaky.

This observation motivated the further evaluation of the two methods in sphere models with a thin skull compartment, where the assumed advantages of the DG-FEM should have a bigger effect. Therefore, we constructed spherical models with a
thinner skull layer, assuming a skull thickness of 2 - 4 mm. The model with the minimal skull thickness of 2 mm, \segresR{2}{2}{82}, has
a skull layer as thin as the edge length of the hexahedrons
(see Figs. \ref{fig:radial-leak}, \ref{fig:vis-current}, \ref{fig:vis-leak}).
Even though a mesh resolution of 1 mm  is strongly recommended for practical
application of the FEM in source analysis \cite{MRull2009,UAydi2014,CHW:Ayd2015},
mesh resolutions of 2 mm are still used even in clinical evaluations \cite{Birot2014}, and
there are areas such as the temporal bone where the skull thickness is actually only 2 mm or even
less \cite[Table 2]{Kwon2006}, so that this is not an artificial scenario. As expected, the
DG-FEM achieved a clearly higher numerical accuracy in the two models
with the
thinnest skull layers, \segresR{2}{2}{82} and \segresR{2}{2}{83}, whereas the results for model \segresR{2}{2}{84} are comparable for the DG- and CG-FEM (see Figure \ref{fig:radial-leak}).
In the latter model, the ratio of resolution (2 mm) and skull thickness (4 mm) guarantees a sufficient resolution and by this already
prohibits leakages.

To make
the difference between the CG- and DG-FEM in the presence of skull leakage better accessible, we generated
Figures \ref{fig:vis-current} and \ref{fig:vis-leak}.
The skull leakage is clearly visible in both figures for model
\segresR{2}{2}{82} and the CG-FEM as described in the results
section.
There is also a slight difference visible in the CSF in all
three models, which might be explained by the
relatively thin CSF layer. At this resolution  (2 mm CSF thickness, 2 mm segmentation
resolution), the elements of the CSF compartment are no longer completely
connected via faces, but often only via shared vertices (as visible in
Figure \ref{fig:vis-current}, left column), which means that for such
a coarse model, the current is blocked in some regions although in the
real geometry it is not. In this case, the CG-FEM shows slightly better
results, as it allows the current to also flow through a single
vertex, which is physically counterintuitive.
In contrast, the DG-FEM does exactly what one would intuitively expect from a mesh based on this segmentation: it channels the main current through the CSF, but due to
the wrong representation of the CSF in the segmentation it yields slightly wrong currents.
It thereby reduces the usually very strong current in the highly conductive CSF compartment, which might explain the slight advantages of the CG-FEM with regard to numerical accuracy for model
\segresR{2}{2}{84} (see especially the lnMAG in Figure \ref{fig:radial-leak}), which is in agreement with the
strong lnMAG effect of modeling the CSF as shown in
\cite[Figure 4]{JVorw2014}. Still, one has to point out that the wrong representation of the CSF geometry has only a very minor effect, as the current is not completely blocked but only slightly diverted.

The findings for the sphere models were underlined by the results obtained using realistic six-compartment head models with mesh resolutions of 1 and 2 mm (Figure \ref{fig:realistic}). Also in this realistic scenario, the DG-FEM showed higher numerical accuracies than the CG-FEM, especially for the lower mesh resolution of 2 mm. The leakages in model \ci{6}{hex}{2mm} are nearly exclusively found in temporal areas, whereas the source positions are regularly distributed over the whole brain. Thus, only a fraction of the sources is strongly affected by leakage effects, and the observed differences between the DG- and CG-FEM in the realistic head model are not as large as one might assume from the results in model \segresR{2}{2}{82}, where the leakages are regularly distributed over the whole model.

Overall, these results
show the benefits of the newly derived DG-FEM approach and motivate the introduction of this new numerical
approach for solving the EEG forward problem.
Furthermore,
the DG-FEM approach allows for an intuitive interpretation of the
results in the presence of segmentation artifacts, which helps in the
interpretation of simulation results, in particular for clinical
experts.

As we have shown in this study, errors in the approximation of the geometry as a result of insufficient image or segmentation resolution and resulting current leakages might
become significant when using hexahedral meshes.
However, there are possibilities to avoid such errors.
In \cite{SVall2010}, a trilinear immersed finite element method to solve the EEG forward
problem was introduced, which allows the use of structured hexahedral meshes, i.e., the
mesh structure is independent of the physical boundaries. The
interfaces are then represented by level-sets and finally considered
using special basis functions. However, this method is still based on
the CG-FEM formulation, so that the behavior
when the thickness of single compartments lies in the range of the resolution of the
underlying mesh is unclear, especially when both the compartment
boundaries between the CSF and skull (inner skull surface) and skull and skin (outer skull surface) are contained
in one element; it is probable that it suffers from the same problems
as the common CG-FEM in such cases. Unfortunately, no further
in-depth analysis for this approach was performed until
now.
Therefore, we claim to have for the first time presented and evaluated an FEM approach preventing
current leakage through single nodes.
In future investigations,
we intend to further develop the already discussed
cut-cell DG approach for source analysis\mbox{\cite{nuessing2016}}, which has the same advantageous features with regard to the representation of the geometry as the approach presented in \cite{SVall2010},
but additionally the charge
preserving property of the DG-FEM as presented here.

The charge preserving property could also be achieved by certain
implementations of finite volume methods. In \cite{MCook2006}, a
vertex-centered finite volume approach was presented that
shares the advantage that anisotropic conductivities can be treated quite naturally with the here-presented DG-FEM approaches.
However, due to its construction, the vertex-centered approach can also be
affected by unphysical current
flow between high-conducting compartments that touch in single nodes
as seen for the CG-FEM. This problem could be avoided using
a cell-centered finite volume approach.

The evaluation of the computational costs of the DG- and CG-FEM showed a higher computational effort for the DG-FEM for the solving of a single equation system and in consequence for the setup of the transfer matrices (Table \ref{tab:computation}). The solving times scaled linearly with the number of degrees of freedom, which corresponds to the theoretically predicted scaling \cite{Braess}. The computation times for the setup of the right-hand side did not differ significantly between the CG- and DG-FEM.

The computation of both the transfer matrix and the right-hand sides can be easily parallelized by simultaneously solving multiple equation systems and setting up multiple right-hand sides, respectively. This simple parallelization approach achieves an optimal scaling with the number of processors, cores and SIMD-lanes. Already a parallel computation of the transfer matrix on four cores, which can be considered as standard equipment nowadays, would reduce $t_\text{transfer}$ to about 8 h for the DG-FEM and model \ci{6}{hex}{1mm}. This reduction of the computation time makes a practical application feasible, since a computation could be carried out overnight. The use of more powerful equipment, as is available in many facilities, would allow for a further speedup. However, in our experiments the overall computation times were dominated by the setup of the right-hand side, which took twice as long as the transfer matrix setup even for the more costly DG-FEM and model \ci{6}{hex}{1mm}. This is a drawback inherent to the subtraction approach. Its nice theoretical properties, which make it preferrable for a first application with new discretization methodology, come at the cost of a dense and expensive-to-compute right-hand side. For the CG-FEM, it was shown that the setup time for the right-hand side vector can be drastically reduced by the adaptation of the direct source modeling approaches, such as Venant, partial integration, or Whitney, that lead to a sparse instead of a dense right-hand-side vector, as previously discussed. For these approaches, the setup time for a single right-hand-side vector is reduced by up to two magnitudes \mbox{\cite{VorwerkDiss2016}}; a similar speedup can be expected for the DG-FEM. Furthermore, just as for the transfer matrix computation, for the computation of the right-hand sides an optimal speedup by parallelization can be easily achieved.

Finally, since the DG approach allows fulfilling the conservation property of electric
charge also in the discrete case, it is not only attractive for source analysis, but
also for the simulation and optimization of brain stimulation methods such as transcranial direct or alternating current
stimulation (tCS, tDCS, tACS)\mbox{\cite{Dmochowski2011,Sadleir2012,Windhoff2013,nuessing2016,CHW:Wag2014}}
or deep brain stimulation \cite{Budson2007,Schmidt2013}.

\section{Conclusion}
We presented theory and numerical evaluation of the subtraction
finite element method (FEM) approach for EEG forward simulations in the discontinuous Galerkin
framework (DG-FEM). We evaluated the accuracy and convergence of the newly
presented approach in spherical and realistic six-compartment models for different mesh resolutions
and compared it to the frequently used Lagrange or continuous Galerkin (CG-)FEM.
In common sphere models, we found similar accuracies of the two approaches for the higher mesh resolutions,
whereas the DG- outperformed the CG-FEM for lower mesh
resolutions. We further compared the approaches in the special
scenario of a very thin skull layer where ``leakages'' might occur. We
found that the DG approach clearly outperforms the CG-FEM in
these scenarios. We underlined these results using visualizations of the electric current flow. The results for the sphere models were confirmed by those obtained in the realistic six-compartment scenario. The computation times presented in this study can easily be reduced through parallelization. Furthermore, different approaches for the setup of the right-hand side are expected to enable a major speedup without loss of accuracy to make a practical application of DG methods in EEG source analysis feasible.
The DG-FEM approach might therefore complement the CG-FEM to
improve source analysis approaches.

\bibliographystyle{siam}
\bibliography{dg_eeg_forward_paper}

\end{document}